\documentclass[11pt,a4paper]{article}
\usepackage[utf8]{inputenc}
\usepackage[T1]{fontenc}
\usepackage{lmodern}
\usepackage[english]{babel}
\usepackage[margin=2.5cm]{geometry} 
\usepackage{graphicx} 
\usepackage{booktabs} 
\usepackage{natbib} 
\usepackage{amsmath} 
\usepackage{amssymb} 
\usepackage{setspace} 
\usepackage{caption} 
\usepackage[colorlinks=true, citecolor=blue, linkcolor=blue, urlcolor=blue]{hyperref} 
\usepackage{amsmath,amssymb,amsthm,bm,mathtools}
\usepackage{bbm}
\usepackage{booktabs}
\usepackage{array}
\usepackage{multirow}
\usepackage{algorithm}
\usepackage{algpseudocode}
\usepackage{graphicx}
\usepackage{longtable, booktabs, pdflscape}
\usepackage{tikz}
\usepackage{float}
\usetikzlibrary{arrows.meta,positioning,shapes.geometric,calc}
\setcitestyle{round}

\let\oldfootnote\footnote
\renewcommand{\footnote}{\fontsize{9}{11}\selectfont\oldfootnote}

\newcommand{\bbeta}{\bm{\beta}}

\newcommand{\T}{^{\mathsf T}}
\newcommand{\given}{\mid}
\DeclareMathOperator*{\argmin}{arg\,min}
\newcommand{\E}{\mathbb{E}}
\newcommand{\prob}{\mathbb{P}}
\newcommand{\R}{\mathbb{R}}
\newcommand{\N}{\mathcal{N}}
\newcommand{\obs}{\mathrm{obs}}
\newcommand{\mis}{\mathrm{mis}}
\newcommand{\mcM}{\mathcal{M}}

\newcommand{\supp}{\mathrm{supp}}

\newtheorem{assumption}{Assumption}
\newtheorem{theorem}{Theorem}

\theoremstyle{remark}

\title{Substantive-Model-Compatible Multiple Imputation for Cox Regression with a Diverging Number of Covariates}

\author{
    Zhilin Zhang \\
    Department of Biostatistics, University of Michigan, Ann Arbor \\
    \and
    Yi Li  \thanks{Corresponding author: yili@umich.edu} \\
     Department of Biostatistics, University of Michigan, Ann Arbor 
}

\date{} 

\begin{document}

\maketitle

\begin{abstract}
Modern biomedical survival studies with high-dimensional genomic and clinical predictors are  challenged by missing covariates. Existing methods conduct inference through penalization and debiasing when the number of covariates diverges with sample size, but they are typically developed with fully observed covariates. Conversely, substantive-model-compatible multiple imputation methods, particularly substantive-model-compatible fully conditional specification (SMC-FCS), provide principled handling of missing covariates while preserving compatibility with the Cox model, yet current methodology and theory remain largely restricted to fixed-dimensional settings.
To address these limitations, we propose a semiparametric multiple imputation framework for inference in  Cox regression with missing covariates of a diverging dimension. Missing covariates are imputed through a high-dimensional SMC-FCS procedure driven by Cox-model likelihood contributions, with rejection sampling used to enforce substantive-model compatibility and ridge-regularized  posterior draws used to stabilize the imputation models. The algorithm  stabilizes the Cox estimator through an imputation-regularized optimization iteration and then generates multiply imputed datasets from a stabilized  chain. Inference for low-dimensional linear functionals or contrasts, $c^\top \bbeta$,  is obtained by combining debiased estimators and within-imputation variance estimates through Rubin's rules.
We establish consistency and asymptotic normality of the resulting pooled estimator under a diverging-dimensional regime. Simulation studies demonstrate favorable finite-sample performance, and an application to the Boston Lung Cancer Survival Cohort illustrates the practical utility of the proposed method for high-dimensional survival studies with incomplete covariates.

\end{abstract}




\section{Introduction}

As a motivating example, we consider the Harvard School of Public Health (HSPH)  subgroup of the Boston Lung Cancer Survival Cohort (BLCS), a large prospective study of  lung cancer survivors investigating cancer genetics. The analytical population consists of 984 European-ancestry patients with pathologically confirmed non-small cell lung cancer (NSCLC), germline genotype measurements, and  survival information. The analysis includes 53 target single nucleotide polymorphisms (SNPs) and 10 clinical covariates, with the goal of identifying factors associated with overall survival. Although the overall missing-cell proportion is only 1.17\%, and no variable has more than 5\% missingness, the accumulated missing-data pattern is substantial: complete-case analysis would retain fewer than 441 patients, eliminating more than half of the cohort. More broadly, modern biomedical survival studies increasingly combine censored outcomes with high-dimensional molecular and clinical predictors, requiring methods that simultaneously address regularized estimation, incomplete covariates, and valid statistical inference.
\vspace{1ex}

Substantial progress has been made for Cox regression \citep{cox1972regression} with a diverging number of covariates ($p \to \infty$ with $p<n$). In these settings, the standard partial likelihood estimator becomes unstable, motivating penalized methods such as the LASSO \citep{tibshirani1997lasso} and SCAD \citep{fan2002variable} for estimation and variable selection. These methods introduced shrinkage bias and  debiasing or desparsification methods were  developed to remove biases due to  penalization and enable confidence intervals and hypothesis testing for low-dimensional regression components \citep{zhang2014confidence,vandegeer2014desparsifying}.  For survival outcomes, \citet{xia2022statistical} established an  inferential method for Cox models with diverging covariate dimension. Unlike earlier debiased methods \citep{fang2017decorrelated,kong2021debiased}, which relied on sparsity assumptions on the inverse information matrix, their projection-based approach remains valid under strong covariate dependence and dense information structures.
These methods are typically developed under the idealized assumption that all covariates are fully observed. As seen in our motivating example, missingness is common in clinical and epidemiologic studies. Ad hoc approaches such as complete-case analysis or single imputation can induce bias, reduce efficiency, and fail to account for uncertainty in the missing values \citep{little2019statistical}.
\vspace{1ex}

A parallel literature addresses missing covariates in survival analysis through Multiple Imputation (MI) \citep{rubin1987multiple}. Standard MI approaches, including Fully Conditional Specification (FCS) \citep{vanbuuren2006fully,white2011multiple}, provide flexibility for mixed data types but often fail to preserve compatibility with the substantive Cox model. This incompatibility, or lack of congeniality \citep{meng1994multiple},  biases hazard ratio estimation \citep{white2009imputing}.
To address this issue, \citet{bartlett2015multiple}, \citet{keogh2018case}, and \citet{keogh2021missing} proposed Substantive-Model-Compatible Fully Conditional Specification (SMC-FCS), which explicitly enforces compatibility by constructing imputations consistent with the  Cox model, providing a principled way for preserving valid Cox regression inference.
\vspace{1ex}

The existing SMC-FCS methodology and theory are largely restricted to fixed-dimensional settings. The imputation models, computational procedures, and theoretical guarantees underlying SMC-FCS do not directly extend to regimes where the number of covariates diverges with sample size. In particular, Rubin’s rules rely on classical large-sample approximations that may break down in these settings, and the interaction between imputation uncertainty and  regularization remains poorly understood.
\vspace{1ex}

Under the diverging settings,  some  work has estimated  models with a diverging number of missing covariates. \citet{wang2015highdim} proposed an EM algorithm with sparsity-enforcing truncation steps and established convergence rates together with asymptotic normality for low-dimensional components. \citet{liang2018imputation} introduced the Imputation-Regularized Optimization (IRO) algorithm, which alternates between imputation and regularized optimization, and established consistency of the averaged estimator under general missingness mechanisms.  However, these works may not apply to  Cox regression.

\vspace{1ex}

 To address this gap, we propose a semiparametric multiple imputation method for inference in Cox regression models with a diverging number of covariates subject to missingness. The proposed method integrates  Cox regression with substantive-model-compatible multiple imputation through a debiased LASSO estimator, where missing covariates are imputed using an SMC-FCS procedure adapted to diverging settings. Imputation is driven by Cox-model likelihood contributions, with rejection sampling used to enforce substantive-model compatibility. Ridge-regularized Bayesian posterior draws stabilize the conditional imputation models when the covariate dimension is large relative to the sample size, and a two-phase algorithm first stabilizes the Cox coefficients through an IRO-type iteration before generating multiply imputed datasets from a thinned posterior chain.  For low-dimensional functionals $c^\top\bbeta$, inference is conducted by applying Rubin's rules to combine the debiased point estimates and within-imputation variance estimators across imputations.

\vspace{1ex}

 The proposed method offers several  advantages. First, it provides a unified approach to handling missing covariates and Cox regression within a coherent inferential framework. Second, by combining debiasing with SMC-FCS, it enables valid inference for scientifically meaningful low-dimensional targets even when the covariate dimension diverges with sample size, whereas existing methods either ignore substantive-model compatibility or fail to propagate imputation uncertainty into inference.  Third, we establish consistency and asymptotic normality of the Rubin-pooled estimator by combining large-sample multiple-imputation theory with high-dimensional Cox inference under the triangular-array regime $p_n\to\infty$ as $n \to\infty$. 
  Finally, the method is computationally feasible and broadly applicable, as illustrated through simulation studies and an analysis of the Boston Lung Cancer Survival Cohort with missing covariates.

\vspace{1ex}

The  paper is organized as follows. Section~\ref{sec:pre} introduces the setup and  the related debiased lasso work. Section~\ref{sec:method} presents the proposed substantive-model-compatible debiased lasso framework, including the iterative imputation-and-optimization procedure, stationary SMC-FCS construction, debiased Cox lasso estimation, and Rubin pooling strategy. Section~\ref{sec:theory} establishes the main theoretical results. Sections~\ref{sec:simulation} and~\ref{sec:data} present simulation studies and a lung cancer survival application, respectively. Section~\ref{sec:disc} concludes with discussion and future directions. All proofs are given in the Appendix.

\section{Preamble}
\label{sec:pre}

\subsection{Setup and Notation}
\label{sec:setup}
 
 Consider a clinical study with $n$ independent subjects, where the $i$th subject contributes $(Y_i, \Delta_i, X_i, R_i)$ for $i = 1, \ldots, n$. The observed time is $Y_i = \min(T_i, C_i)$, where $T_i$ denotes the failure time and $C_i$ an independent right-censoring time, with event indicator $\Delta_i = I(T_i \le C_i)$. The covariate vector $X_i \in \R^p$ may be partially observed, with $R_i = (R_{i1}, \ldots, R_{ip}) \in \{0,1\}^p$ denoting the missingness pattern, where $R_{ij} = 1$ if $X_{ij}$ is observed and $0$ otherwise. We assume  that $T_i$ follows the Cox proportional hazards model
\[
\lambda(t \given X_i) = \lambda_0(t)\exp(X_i\T \beta^0), \qquad t \ge 0,
\]
where $\lambda_0(\cdot)$ is an unspecified baseline hazard and $\beta^0 \in \R^p$ is the regression parameter of interest. Also denote by $\Lambda_0(t) = \int_0^t \lambda_0(s)ds$ the cumulative baseline hazard. 
We consider a diverging dimensional regime in which the covariate dimension
$p=p_n$ diverges with the sample size $n$, with
\(
p_n \to \infty,
\,
p_n <n,
\)
while the true regression coefficient $\beta^0$ remains sparse with
$s_0=\|\beta^0\|_0 \ll n$.

\vspace{1ex}
For each subject, we partition the covariates into observed and missing components as $X_i = (X_i^{\obs}, X_i^{\mis})$, where $X_i^{\obs} = \{X_{ij}: R_{ij}=1\}$ and $X_i^{\mis} = \{X_{ij}: R_{ij}=0\}$. Let $\mcM = \{j : \sum_{i=1}^n (1-R_{ij}) > 0\}$ denote the set of partially observed variables. Missingness is assumed to be missing at random (MAR), namely $P(R_i \mid X_i, Y_i, \Delta_i) = P(R_i \mid X_i^{\obs}, Y_i, \Delta_i)$, so that the missingness mechanism is ignorable for likelihood-based inference.

\vspace{1ex}
We refer to a \emph{completed covariate matrix} as $X = (X_1, \ldots, X_n)\T \in \R^{n \times p}$. Let $\ell_n(\beta;X)$ denote the negative Cox partial log-likelihood normalized by $n$, i.e.,
\[
\ell_n(\beta;X)
=
-
\frac1n
\sum_{i=1}^n
\Delta_i
\left[
X_i^\top\beta
-
\log
\left\{
\sum_{j:Y_j\ge Y_i}
\exp(X_j^\top\beta)
\right\}
\right],
\]
and let $\dot\ell_n(\beta;X)$ denote its gradient with respect to $\beta$. Define the population information matrix as $$\Sigma_{\beta^0} = \E[\{X_i-\eta_0(Y_i;\beta^0)\}^{\otimes 2}\Delta_i],$$ where
\[
\eta_0(t;\beta^0)
=
\frac{
\E\!\left[
X_i \exp(X_i\T\beta^0) I(Y_i \ge t)
\right]
}{
\E\!\left[
\exp(X_i\T\beta^0) I(Y_i \ge t)
\right]
}
\]
is the population risk set weighted covariate mean at time $t$. Its sample analogue is
\begin{equation} \label{sigma}
\widehat\Sigma_{\beta}
=
\frac{1}{n}
\sum_{i=1}^n
\Delta_i
\bigl\{
X_i-\widehat\eta_n(Y_i;\beta,X)
\bigr\}^{\otimes 2},
\end{equation}
with
\(
\widehat\eta_n(t;\beta,X)
=
\frac{
\sum_{i=1}^n
X_i \exp(X_i\T\beta) I(Y_i \ge t)
}{
\sum_{i=1}^n
\exp(X_i\T\beta) I(Y_i \ge t)
}.
\)

\vspace{1ex}

Let $\Theta_{\beta^0}=\Sigma_{\beta^0}^{-1}$ denote the inverse population information matrix. The main inferential targets are low-dimensional linear functionals such as $c^\top\beta^0$, where $c\in\R^p$ is a fixed loading vector.

\subsection{Review of debiased lasso inference with no missingness \citep{xia2022statistical}}
 
 When  $X$ is fully observed, the Cox lasso estimator is
\[
\widehat\beta_{L_1}
=
\argmin_{\beta \in \R^p}
\bigl\{
\ell_n(\beta; X) + \lambda_n\|\beta\|_1
\bigr\}.
\]
 As penalization introduces shrinkage bias, confidence intervals based directly on $\widehat\beta_{L_1}$ are invalid.
To remove this bias, \citet{xia2022statistical} proposed the debiased estimator
\[
\widehat b
=
\widehat\beta_{L_1}
-
\widehat\Theta\, \dot\ell_n(\widehat\beta_{L_1}; X),
\]
where $\widehat\Theta \in \R^{p \times p}$ estimates the inverse information matrix $\Theta_{\beta^0} = \Sigma_{\beta^0}^{-1}$. A key feature of their construction is the estimation of $\widehat\Theta$;  \citet{xia2022statistical} replaced $\ell_0$-sparsity assumptions with an $\ell_1$-constraint and estimate each row of $\widehat\Theta$ by solving
\begin{equation}
\label{eq:xia-qp}
\widehat\Theta_{j\cdot}
=
\argmin_{m \in \R^p}
m\T \widehat\Sigma_{\widehat\beta_{L_1}}\, m
\quad
\text{subject to}
\quad
\bigl\|
\widehat\Sigma_{\widehat\beta_{L_1}}\, m - e_j
\bigr\|_\infty
\le
\gamma_n,
\qquad
j = 1, \ldots, p,
\end{equation}
where $e_j$ is the $j$th canonical basis vector and
\( \widehat\Sigma_{\widehat\beta_{L_1}}
\) is as defined in \eqref{sigma}.
The estimator $\widehat\Theta$ is then assembled row by row from the $p$ optimization problems in \eqref{eq:xia-qp}.

\vspace{1ex}

For inference on a fixed contrast $c\T\beta^0$ with $\|c\|_2 = 1$ and $\|c\|_1 \le a_\ast < \infty$, Theorem~1 of  \citep{xia2022statistical} establishes
\[ 
\sqrt{n}\, c\T\bigl(\widehat b - \beta^0\bigr) \big/ \sqrt{c\T \widehat\Theta\, c} \overset{d}{\longrightarrow} \N(0, 1),
\]
under bounded covariates, bounded eigenvalues of $\Sigma_{\beta^0}$, the tuning rates $\lambda_n \asymp \sqrt{\log p / n}$ and $\gamma_n \asymp \|\Theta_{\beta^0}\|_{1,1} s_0 \lambda_n$, and the rate condition $\|\Theta_{\beta^0}\|_{1,1}^2\, p\, s_0\, \log(p)/\sqrt{n} \to 0$. The $(1,1)$-induced operator norm of a $p \times p$ matrix $A=(a_{ij})$ is defined as $\|A\|_{(1,1)}=\max_{1\le j\le p}\sum_{i=1}^p |a_{ij}|$.
The factor $\|\Theta_{\beta^0}\|_{1,1}$ in the rate condition replaces the $\ell_0$-sparsity assumption that earlier debiased-lasso constructions imposed: the inverse information matrix is allowed to be dense provided that its $(1,1)$-norm grows slowly enough relative to $n$. Variance estimation uses $c\T \widehat\Theta\, c / n$, and a standard normal reference distribution yields confidence intervals.

\section{Proposed Substantive-Model-Compatible Debiased Lasso}
\label{sec:method}

 A major challenge in survival analysis with missing covariates is preserving compatibility between the imputation procedure and the substantive Cox model while accommodating high-dimensional regularization and valid post-selection inference. To address this, we propose substantive-model-compatible debiased lasso (SMC-DBL) that integrates substantive-model-compatible multiple imputation, penalized Cox regression, and debiased inference. The procedure iteratively updates the imputations and substantive Cox model until approximate compatibility is achieved. Specifically, missing covariates are updated through an SMC-FCS mechanism anchored to the Cox likelihood, while the regression coefficient and baseline cumulative hazard are repeatedly re-estimated from the evolving completed data. After convergence and burn-in, independently generated completed datasets are analyzed using the debiased Cox lasso  \citep{xia2022statistical}, and the resulting estimators are combined using Rubin's rules. Additional implementation details are summarized in Algorithm~\ref{alg:smcdbl}. 

 \subsection{SMC-FCS sweeps}
\label{sec:sweep}

Fix a missing variable $j \in \mcM$ and condition on the current completed design $X$, the current Cox coefficient $\beta$, and the current baseline cumulative hazard estimator $\widehat\Lambda_0$. Let $\obs_j = \{i : R_{ij} = 1\}$ index the rows on which $X_j$ is observed, set $n_{\obs,j} = |\obs_j|$, write $y_{\obs,j} = (X_{ij})_{i \in \obs_j}$, and let $Z_{\obs,j}$ be the row-centred submatrix of $X_{-j}$ restricted to $\obs_j$.  We handle continuous and discrete covariates separately.

 \vspace{1ex}
 
For continuous $X_j$, the working model is Gaussian. With ridge multiplier $\lambda_{\mathrm{ridge}}$ chosen as in Section~\ref{sec:tuning}, the ridge point estimate and residual variance are
\[
\widehat\gamma_j
=
\left(
\frac{Z_{\obs,j}\T Z_{\obs,j}}{n_{\obs,j}}
+
\lambda_{\mathrm{ridge}} I
\right)^{-1}
\frac{Z_{\obs,j}\T y_{\obs,j}}{n_{\obs,j}},
\qquad
\widehat\sigma^2_j
=
\frac{1}{\mathrm{df}}
\bigl\|
y_{\obs,j}
-
\widehat\alpha_j
-
Z_{\obs,j}\widehat\gamma_j
\bigr\|_2^2,
\]
with $\mathrm{df} = \max(n_{\obs,j}-p,1)$. A Bayes-type draw of the working-model parameters is then obtained via
\[
\sigma^{2,\star}_j
\sim
\frac{\mathrm{df}\,\widehat\sigma^2_j}
{\chi^2_{\mathrm{df}}},
\qquad
\gamma^\star_j
\sim
\N\!\left(
\widehat\gamma_j,\,
\frac{\sigma^{2,\star}_j}{n_{\obs,j}}
\left(
\frac{Z_{\obs,j}\T Z_{\obs,j}}{n_{\obs,j}}
+
\lambda_{\mathrm{ridge}} I
\right)^{-1}
\right),
\]
and the corresponding intercept $\alpha^\star_j$ is reconstructed from the centring constants.

\vspace{1ex}

For each subject $i$ with $R_{ij}=0$, a proposal
$x^\star
\sim
\N(
\alpha^\star_j + X_{i,-j}\T\gamma^\star_j,
\sigma^{2,\star}_j)$
is drawn and truncated to $[-K,K]$. Let $x_{\mathrm{curr}}$ denote the current imputed value. 
The proposal is accepted with probability
\begin{equation}
\label{eq:mhratio}
\alpha(x_{\mathrm{curr}},x^\star)
=
\min\!\left\{
1,\,
\frac{
f_{\mathrm{Cox}}
(Y_i,\Delta_i\mid x^\star,X_{i,-j};
\beta,\widehat\Lambda_0)
}{
f_{\mathrm{Cox}}
(Y_i,\Delta_i\mid x_{\mathrm{curr}},X_{i,-j};
\beta,\widehat\Lambda_0)
}
\right\},
\end{equation}
where
\[
f_{\mathrm{Cox}}(Y_i,\Delta_i\mid X_i;\beta,\widehat\Lambda_0)
=
\{
\Delta\widehat\Lambda_0(Y_i)\exp(X_i^\top\beta)
\}^{\Delta_i}
\exp\!\left[
-
\widehat\Lambda_0(Y_i)\exp(X_i^\top\beta)
\right],
\]
with
\(
\Delta\widehat\Lambda_0(Y_i)
=
\widehat\Lambda_0(Y_i)-\widehat\Lambda_0(Y_i-)
\)
denoting the jump of the Breslow baseline cumulative hazard estimator at the observed failure time. For censored observations, it follows that 
$\Delta\widehat\Lambda_0(Y_i)=0$ with
$\Delta_i=0$, and we follow the convention of $0^0=1$.  

\vspace{1ex}

The accept-reject rule in equation~\eqref{eq:mhratio} is the Metropolis-Hastings analogue of the rejection-sampling acceptance ratio used by~\cite{bartlett2015multiple}, with the Gaussian working model serving as the proposal distribution.

\vspace{1ex}

For binary, ordinal categorical, or unordered categorical $X_j$, the working model is logistic, proportional-odds, or multinomial logistic, respectively, each fitted by ridge-penalized maximum likelihood using the observed rows. Bayesian draws of the working-model parameters are obtained from a Laplace approximation centred at the ridge estimator, with covariance given by the inverse regularized observed information matrix. Candidate draws $X_{ij}^\star$ for missing entries are then sampled from the corresponding conditional distribution, and the Cox-likelihood ratio in equation~\eqref{eq:mhratio} again determines acceptance.

\subsection{Substantive-model Updates} \label{sub-update}

  In our proposed procedure, an \emph{inner sweep} is one full pass of the SMC-FCS chain over all incomplete covariates conditional on the current $(\beta,\widehat\Lambda_0)$, while an \emph{outer sweep} is one complete IRO iteration consisting of $S_{\mathrm{in}}$ inner sweeps followed by updating the regression coefficient and baseline cumulative hazard estimates using the updated completed dataset. Algorithmically, the procedure uses two levels of indexing: $\ell=0,1,2,\ldots$ indexes the outer IRO iterations within a chain, while $m=1,\ldots,M$ indexes the independently initialized inferential chains (which will be later combined through Rubin pooling).
  
\vspace{1ex}

 Specifically, for chain $m$ at outer iteration $\ell$, given the current model parameters
\(
(\beta_m^{(\ell)}, \widehat\Lambda_{0m}^{(\ell)})
\),
run $S_{\mathrm{in}}$ inner sweeps as described in Section~\ref{sec:sweep}, and let $\widetilde X_m^{(\ell+1)}$ denote the completed dataset obtained after the final inner sweep. We update the estimate of the regression coefficients  by
\[
\beta_m^{(\ell+1)}
=
\argmin_\beta
\left\{
\ell_n(\beta;\widetilde X_m^{(\ell+1)})
+
\lambda_n\|\beta\|_1
\right\},
\]
and the corresponding Breslow baseline cumulative hazard estimator is updated by
\[
\widehat\Lambda_{0m}^{(\ell+1)}(t)
=
\sum_{r:Y_r\le t,\Delta_r=1}
\left[
\sum_{i:Y_i\ge Y_r}
\exp\left\{
(\widetilde X_{im}^{(\ell+1)})^\top
\beta_m^{(\ell+1)}
\right\}
\right]^{-1}.
\]
The outer IRO iterations continue until
\[
\|\beta_m^{(\ell+1)}-\beta_m^{(\ell)}\|_1
+
\|\widehat\Lambda_{0m}^{(\ell+1)}
-
\widehat\Lambda_{0m}^{(\ell)}\|_\infty
<
\varepsilon,
\]
where $\varepsilon>0$ is a pre-specified tolerance,  or the maximum number of iteration steps is reached. 

At the final iteration step,  set
    $(\beta^\star_m,
    \widehat\Lambda^\star_{0m},
    \widetilde X^\star_m)
    =
    (\beta^{(\ell+1)}_m,
    \widehat\Lambda^{(\ell+1)}_{0m},
    \widetilde X^{(\ell+1)}_m)$.

 \subsection{Per-chain debiased lasso and pooling}
 Starting from $\widetilde X^\star_m$, run $T_0$ additional SMC-FCS sweeps conditional on fixed $(\beta^\star_m,\widehat\Lambda^\star_{0m})$ and ridge level $\widehat\lambda_{\mathrm{ridge}}$.  We retain the terminal completed dataset as $\widetilde X^{(m)}$, and  fit the Cox lasso on $\widetilde X^{(m)}$:
    $$\widehat\beta^{(m)}_{L_1}
    =
    \argmin_\beta
    \{
    \ell_n(\beta;\widetilde X^{(m)})
    +
    \widehat\lambda_n\|\beta\|_1
    \}.$$
 Let
\(
\widehat\Sigma^{(m)}
=
\ddot\ell_n^{(m)}
\bigl(
\widehat\beta_{L_1}^{(m)}
\bigr)
\)
denote the observed Hessian of the Cox partial log-likelihood evaluated at $\widehat\beta_{L_1}^{(m)}$. An approximate inverse information matrix $\widehat\Theta^{(m)}$ is constructed row by row by solving the nodewise quadratic programme
\begin{equation}
\label{eq:qp}
\widehat\Theta_k^{(m)}
=
\argmin_{u\in\mathbb R^p}
u^\top\widehat\Sigma^{(m)}u
\quad
\text{subject to}
\quad
\|
\widehat\Sigma^{(m)}u-e_k
\|_\infty
\le
\gamma_n,
\qquad
k=1,\ldots,p,
\end{equation}
where $e_k$ is the $k$th canonical basis vector. In practice, equation~\eqref{eq:qp} is solved on the positive-eigenvalue subspace of $\widehat\Sigma^{(m)}$, with a ridge fallback used when infeasibility is encountered.

The per-chain debiased estimator is
\[
\widehat\beta_{\mathrm{db}}^{(m)}
=
\widehat\beta_{L_1}^{(m)}
-
\widehat\Theta^{(m)}
\dot\ell_n^{(m)}
\bigl(
\widehat\beta_{L_1}^{(m)}
\bigr).
\]

Pooling across the $M$ retained chains follows Rubin's rules. Define
\[
\bar\beta_M
=
\frac1M
\sum_{m=1}^M
\widehat\beta_{\mathrm{db}}^{(m)},
\qquad
\widehat V_W
=
\frac1M
\sum_{m=1}^M
\frac{\widehat\Theta^{(m)}}{n},
\]
\[
\widehat V_B
=
\frac1{M-1}
\sum_{m=1}^M
\left(
\widehat\beta_{\mathrm{db}}^{(m)}
-
\bar\beta_M
\right)
\left(
\widehat\beta_{\mathrm{db}}^{(m)}
-
\bar\beta_M
\right)^\top,
\qquad
\widehat V_{\mathrm{total}}
=
\widehat V_W
+
\left(
1+\frac1M
\right)
\widehat V_B.
\]

The resulting $(1-\alpha)$ confidence interval for the $k$th component is
\[
\bar\beta_{M,k}
\pm
z_{\alpha/2}
\sqrt{
\widehat V_{\mathrm{total},kk}
},
\qquad
k=1,\ldots,p,
\]
where $z_{\alpha/2}$ is the upper $\alpha/2$ standard normal quantile.

\vspace{1ex}
The standard normal reference, rather than a Barnard--Rubin $t$ approximation, is justified by the asymptotic regime considered in Theorem~\ref{thm:clt}.  Inference for a fixed-dimensional linear contrast $c^\top\beta^0$ proceeds by replacing $\bar\beta_M$ with $c^\top\bar\beta_M$ and $\widehat V_{\mathrm{total}}$ with $c^\top\widehat V_{\mathrm{total}}c$ and applying the corresponding normal approximation; see  Theorem~\ref{thm:clt}.
\subsection{Tuning of the hyper-parameters}
\label{sec:tuning}
 
Two hyper-parameters enter the procedure. The Cox lasso penalty $\lambda_n$ is selected by five-fold cross-validation on the partial likelihood at each Cox fit, which delivers the rate $\lambda_n \asymp \sqrt{\log p / n}$ required by \citet{huang2013oracle}. The ridge penalty is fixed at
\(
\lambda_{\mathrm{ridge}} \;=\; \lambda_n^2 (p-1),
\)
a scaling that aligns with the ridge prediction consistency requirement of Assumption~\ref{ass:ridge} and that admits the closed-form bias-variance decomposition used in Section~\ref{sec:theory}. The debiased-lasso constraint radius is set to $\gamma_n = a\,\sqrt{\log p / n}$ with a fixed multiplier $a = 0.5$, an interior value located in a three-point sensitivity analysis over $(a = 0.25, 0.5, 1.0)$. Empirical exploration of cross-validation surfaces for $a$ on thresholded debiased estimates produced essentially flat objectives across the interior of any reasonable grid, which is consistent with the order-one scaling implied by \citet{xia2022statistical}.

 \begin{algorithm}[H]
\caption{SMC-DBL for Cox Regression with Iterated Baseline Hazard}
\label{alg:smcdbl}
\begin{algorithmic}[1]

\Require Observed data $\{(Y_i,\Delta_i,X_i^{\obs},R_i)\}_{i=1}^n$; number of imputations $M$; inner sweeps $S_{\mathrm{in}}$; burn-in length $T_0$; outer tolerance $\varepsilon$; CV folds $K_{\mathrm{cv}}$; nominal ridge level $\lambda_{\mathrm{ridge}}^{(0)}$.
\Ensure Pooled debiased estimator $\bar\beta_M$, variance estimator $\widehat V_{\mathrm{total}}$, and confidence intervals.

\Statex
\State \textbf{Phase 0: Tuning.}

\State Initialise missing entries by random draws from observed values within each incomplete covariate.

\State Run a preliminary IRO loop with inner SMC-FCS sweeps, iteratively updating both the Cox lasso coefficient $\beta$ and the Breslow baseline cumulative hazard $\widehat\Lambda_0$, until the joint outer iterates converge. Denote the resulting completed data by $\widetilde X^{(0)}$.

\State Choose $\widehat\lambda_n$ by $K_{\mathrm{cv}}$-fold cross-validation of the Cox partial likelihood on $\widetilde X^{(0)}$.

\State Choose $\widehat b$ by cross-validation for the working conditional models of the incomplete covariates using only rows where the target covariate is observed.

\State Set $\widehat\lambda_{\mathrm{ridge}}=\widehat b\,p\log(p)/n$ and $\gamma_n=a\sqrt{\log p/n}$.

\State Discard the preliminary completed dataset $\widetilde X^{(0)}$.

\Statex
\State \textbf{Phases 1--3: Independent inferential chains.}

\For{$m=1,\ldots,M$}

    \Statex \quad \textbf{Phase 1.$m$: IRO convergence.}

    \State Initialise missing entries independently  and form the completed dataset $\widetilde X^{(0)}_m$.

    \State Fit the Cox lasso:
    $\beta^{(0)}_m
    =
    \argmin_\beta
    \{
    \ell_n(\beta;\widetilde X^{(0)}_m)
    +
    \widehat\lambda_n\|\beta\|_1
    \}$.

    \State Compute  
    $$\widehat\Lambda^{(0)}_{0m}(t)
    =
    \sum_{r:Y_r\le t,\Delta_r=1}
    \Big[
    \sum_{i:Y_i\ge Y_r}
    \exp\{(\beta^{(0)}_m)^\top \widetilde X^{(0)}_{im}\}
    \Big]^{-1}.$$

    \For{$\ell=0,1,2,\ldots$ }

        \State \emph{I-step.}
        Conditional on
        $(\beta^{(\ell)}_m,\widehat\Lambda^{(\ell)}_{0m})$,
        run $S_{\mathrm{in}}$ SMC-FCS sweeps: in each sweep,  visit every incomplete $X_j$ once and update its missing entries from a conditional distribution proportional to
        $f_j(x_j\mid X_{-j};\widehat\alpha_j,\widehat\lambda_{\mathrm{ridge}})
        \,
        f_{\mathrm{Cox}}
        \{Y,\Delta\mid X_j=x_j,X_{-j};
        \beta^{(\ell)}_m,
        \widehat\Lambda^{(\ell)}_{0m}\}$.   Retain the terminal completed dataset and denote it  by $\widetilde X^{(\ell+1)}_m$.

        \State \emph{RO-step.}
        Refit the Cox lasso:
        $\beta^{(\ell+1)}_m
        =
        \argmin_\beta
        \{
        \ell_n(\beta;\widetilde X^{(\ell+1)}_m)
        +
        \widehat\lambda_n\|\beta\|_1
        \}$.

        \State Update 
        $$\widehat\Lambda^{(\ell+1)}_{0m}(t)
        =
        \sum_{r:Y_r\le t,\Delta_r=1}
        \Big[
        \sum_{i:Y_i\ge Y_r}
        \exp\{(\beta^{(\ell+1)}_m)^\top
        \widetilde X^{(\ell+1)}_{im}\}
        \Big]^{-1}.$$

        \State Stop if
        $\|\beta^{(\ell+1)}_m-\beta^{(\ell)}_m\|_1
        +
        \sup_{t\le\tau}
        |
        \widehat\Lambda^{(\ell+1)}_{0m}(t)
        -
        \widehat\Lambda^{(\ell)}_{0m}(t)
        |
        <
        \varepsilon$.

    \EndFor

    \State Set
    $(\beta^\star_m,
    \widehat\Lambda^\star_{0m},
    \widetilde X^\star_m)
    =
    (\beta^{(\ell+1)}_m,
    \widehat\Lambda^{(\ell+1)}_{0m},
    \widetilde X^{(\ell+1)}_m)$.

\algstore{smcdblbreak}

\end{algorithmic}
\end{algorithm}

\begin{algorithm}[H]
\begin{algorithmic}[1]

\algrestore{smcdblbreak}

    \Statex \quad \textbf{Phase 2.$m$: Burn-in.}

    \State Starting from $\widetilde X^\star_m$, run $T_0$ additional SMC-FCS sweeps conditional on fixed $(\beta^\star_m,\widehat\Lambda^\star_{0m})$ and ridge level $\widehat\lambda_{\mathrm{ridge}}$.

    \State Retain the terminal completed dataset as $\widetilde X^{(m)}$.

    \Statex \quad \textbf{Phase 3.$m$: Debiased lasso inference.}

    \State Fit the Cox lasso on $\widetilde X^{(m)}$:
    $\widehat\beta^{(m)}_{L_1}
    =
    \argmin_\beta
    \{
    \ell_n(\beta;\widetilde X^{(m)})
    +
    \widehat\lambda_n\|\beta\|_1
    \}$.

    \State Solve the nodewise quadratic programs with constraint radius $\gamma_n$ to obtain $\widehat\Theta^{(m)}$.

    \State Compute the debiased estimator
    $\widehat\beta^{(m)}_{\mathrm{db}}
    =
    \widehat\beta^{(m)}_{L_1}
    -
    \widehat\Theta^{(m)}
    \dot\ell_n
    (\widehat\beta^{(m)}_{L_1};
    \widetilde X^{(m)})$.

    \State Estimate the corresponding within-imputation covariance matrix $\widehat V^{(m)}$.

\EndFor

\Statex
\State \textbf{Phase 4: Rubin pooling.}

\State Compute the pooled estimator
$\bar\beta_M
=
M^{-1}
\sum_{m=1}^M
\widehat\beta^{(m)}_{\mathrm{db}}$.

\State Compute the within-imputation covariance
$\widehat V_W
=
M^{-1}
\sum_{m=1}^M
\widehat V^{(m)}$.

\State Compute the between-imputation covariance
$\widehat V_B
=
(M-1)^{-1}
\sum_{m=1}^M
(\widehat\beta^{(m)}_{\mathrm{db}}-\bar\beta_M)
(\widehat\beta^{(m)}_{\mathrm{db}}-\bar\beta_M)^\top$.

\State Compute the Rubin variance estimator
$\widehat V_{\mathrm{total}}
=
\widehat V_W
+
(1+M^{-1})
\widehat V_B$.

\State Construct confidence intervals using the diagonal entries of $\widehat V_{\mathrm{total}}$.

\end{algorithmic}
\end{algorithm}

\section{Theoretical Results}
\label{sec:theory}
 
 We begin with the assumptions required for the asymptotic analysis. Throughout, the covariate dimension is allowed to diverge with the sample size, and we write $p=p_n$.

\begin{assumption}
\label{ass:dimension}
The covariate dimension satisfies 
\(
p_n=O(n^\kappa)
\)
for some fixed $\kappa\in(0,1)$. All stochastic orders are understood under this triangular-array regime.
\end{assumption}

\begin{assumption}
\label{ass:bounded}
There exists $K<\infty$, independent of $(n,p_n)$, such that $\|X_i\|_\infty\le K$ almost surely for all $i=1,\ldots,n$.
\end{assumption}

\begin{assumption}
\label{ass:linearpred}
There exists $K_1<\infty$, independent of $(n,p_n)$, such that $|X_i^\top\beta^0|\le K_1$ uniformly for all $i=1,\ldots,n$ almost surely.
\end{assumption}

\begin{assumption}
\label{ass:followup}
The maximum follow-up time satisfies $0<\tau<\infty$ and $\pi_0=\prob(Y\ge\tau)>0$.
\end{assumption}

\begin{assumption}
\label{ass:baseline}
The true baseline cumulative hazard $\Lambda_0^0$ is continuous and nondecreasing on $[0,\tau]$, and satisfies $\Lambda_0^0(\tau)\le C_\Lambda<\infty$.
\end{assumption}

\begin{assumption}
\label{ass:eigen}
The population information matrix
\(
\Sigma_{\beta^0}
=
\E\bigl[
\{X_i-\eta_0(Y_i;\beta^0)\}^{\otimes2}\Delta_i
\bigr]
\)
satisfies
\(
0<\zeta_{\min}
\le
\zeta_{\min}(\Sigma_{\beta^0})
\le
\zeta_{\max}(\Sigma_{\beta^0})
\le
\zeta_{\max}<\infty.
\)
\end{assumption}

\begin{assumption}
\label{ass:working}
For each partially observed covariate $X_j$, let $q_{ij}(\cdot)$ denote the working conditional imputation distribution for subject $i$, chosen according to the data type of $X_j$:
\begin{enumerate}
\item If $X_j$ is continuous, $q_{ij}(\cdot)$ is a Gaussian density with ridge-estimated mean $\hat\mu_{ij}$ and variance $\hat\sigma_j^2$, truncated to $[-K,K]$.

\item If $X_j$ is binary with support $\{0,1\}$, $q_{ij}(\cdot)$ is a Bernoulli mass function with success probability $\hat\pi_{ij}=\Pr(X_{ij}=1\mid X_{i,-j})$, estimated by ridge-penalized logistic regression.

\item If $X_j$ is an
ordered categorical variable, e.g., an additively coded SNP with support $\{0,1,2\}$, $q_{ij}(\cdot)$ is an ordinal logistic mass function with probabilities $\hat\pi_{ij}^{(0)},\hat\pi_{ij}^{(1)},\hat\pi_{ij}^{(2)}$, estimated by ridge-penalized proportional-odds regression.

\item If $X_j$ is categorical with $k_j\ge3$ unordered levels, $q_{ij}(\cdot)$ is a categorical mass function with probabilities $\hat\pi_{ij}^{(\ell)}$, estimated by ridge-penalized multinomial logistic regression.
\end{enumerate}
In all cases, there exist constants $0<c_\theta\le C_\theta<\infty$, independent of $(n,p_n,i,j)$, such that $c_\theta\le q_{ij}(x)\le C_\theta$ for all $x$ in the support of $X_j$. For discrete working models, these bounds are enforced by numerical clipping, $\hat\pi_{ij}^{(\ell)}\in[\epsilon,1-\epsilon]$, for some fixed $\epsilon>0$.
\end{assumption}

\begin{assumption}
\label{ass:missing}
Let $d_n=|\{(i,j):R_{ij}=0\}|$ denote the total number of missing entries and let $n_{\mis}=|\{i:\sum_{j=1}^{p_n}(1-R_{ij})>0\}|$ denote the number of subjects with at least one missing covariate. The following hold:
\begin{enumerate}
\item $d_n\le \bar r\,n p_n$ for some $\bar r\in(0,1)$ bounded away from $1$.
\item $n_{\mis}/n\le \bar r_{\mathrm{sub}}$ for some $\bar r_{\mathrm{sub}}\in(0,1)$.
\item The missingness mechanism is MAR:
\(
\Pr(R_i\mid X_i,Y_i,\Delta_i)
=
\Pr(R_i\mid X_i^{\mathrm{obs}},Y_i,\Delta_i).
\)
\end{enumerate}
\end{assumption}

\begin{assumption}
\label{ass:ridge}
For each partially observed covariate $X_j$, let $\alpha_j^0$ denote the true parameter vector of the corresponding working imputation model, where $\alpha_j^0=\gamma_j^0$ for continuous variables, $\alpha_j^0$ is the logistic or multinomial regression coefficient vector for binary and categorical variables, and $\alpha_j^0$ additionally includes threshold parameters for ordinal logistic models. Let
\(
s_\ast=\max_{j\in\mathcal M}\|\alpha_j^0\|_0 .
\)
Under $p_n=O(n^\kappa)$, assume
\(
s_\ast\log n=o(\sqrt n).
\)
The ridge-regularized working-model estimators satisfy, uniformly over $j\in\mathcal M$,
\(
\|\widehat\alpha_j-\alpha_j^0\|_2
=
O_p\!\left(\sqrt{\frac{s_\ast\log n}{n}}\right),
\)
and the corresponding prediction error satisfies
\(
\max_{j\in\mathcal M}
\frac1n
\sum_{i=1}^n
\{X_{i,-j}^\top(\widehat\alpha_j-\alpha_j^0)\}^2
=
o_p(n^{-1/2}).
\)
\end{assumption}

\begin{assumption}
\label{ass:sparsity}
Let $s_0=\|\beta^0\|_0$, $s_\ast=\max_{j\in\mathcal M}\|\alpha_j^0\|_0$, and $\Theta_{\beta^0}=\Sigma_{\beta^0}^{-1}$. Under $p_n=O(n^\kappa)$ for some $\kappa\in(0,1)$, assume
\(
s_0\log n=o(\sqrt n),
\)
and
\(
\|\Theta_{\beta^0}\|_{1,1}^2
(\log n)\cdot
\max(n^\kappa s_0,s_\ast^2)
=
o(n),
\)
where $\|\cdot\|_{1,1}$ is the matrix operator norm induced by the $\ell_1$ vector norm.
\end{assumption}

\begin{assumption}
\label{ass:Sin}
Let $\varepsilon_a=c_fc_\theta/(C_fC_\theta)>0$ be the one-coordinate minorisation constant from Theorem~\ref{thm:ergodicity}. Since a full SMC-FCS sweep is a composition over $d_n$ missing coordinates, the sweep-level minorisation constant is $\varepsilon_a^{d_n}$. The number of inner sweeps satisfies
\(
S_{\mathrm{in}}
\ge
c_S
\frac{\log(n p_n)}{\varepsilon_a^{d_n}},
\qquad c_S>1.
\)
 
\end{assumption}

\begin{assumption}
\label{ass:Lout}
The number of outer IRO iterations $\ell$ satisfies
\(
\ell
\ge
c_\ell
\frac{\log n}{\log(1/\lambda_\ast)},
\, c_\ell>1,
\)
where $\lambda_\ast\in(0,1)$ is the local stability constant from Assumption~\ref{ass:contract}.
\end{assumption}

\begin{assumption}
\label{ass:contract}
Let $\theta=(\beta,\Lambda_0)$, and let $\mathcal M(\theta)$ denote the population IRO map induced by the stationary SMC-FCS law followed by the Cox lasso and Breslow updates. There exist a fixed point $\theta^\ast=(\beta^\ast,\Lambda_0^\ast)$, a neighbourhood $\mathcal N(\theta^\ast)$, a constant $\lambda_\ast\in(0,1)$, and a deterministic sequence $a_n=o(1)$ such that
\[
\|\mathcal M(\theta)-\theta^\ast\|_{\mathcal H}
\le
\lambda_\ast
\|\theta-\theta^\ast\|_{\mathcal H}
+
a_n,
\qquad
\theta\in\mathcal N(\theta^\ast),
\]
where
\(
\|\theta-\theta^\ast\|_{\mathcal H}
=
\|\beta-\beta^\ast\|_1
+
\|\Lambda_0-\Lambda_0^\ast\|_\infty .
\)
\end{assumption}

\begin{assumption} 
\label{ass:burnin}
After the outer IRO iterates have reached the local neighbourhood of $\theta^\ast$, each retained chain is run for $T_0$ additional SMC-FCS sweeps conditional on the limiting substantive-model parameters. The burn-in length satisfies
\(
T_0
\ge
c_T
\frac{\log(n p_n)}{\varepsilon_a^{d_n}},
\qquad c_T>1.
\)
Thus the retained completed datasets are asymptotically sampled from the stationary imputation distribution associated with $(\beta^\ast,\Lambda_0^\ast)$.
\end{assumption}

 Assumptions~\ref{ass:dimension}--\ref{ass:eigen} are standard regularity conditions for high-dimensional Cox regression under a diverging-dimensional regime with $p_n=O(n^\kappa)$ \citep{xia2022statistical}. Assumption~\ref{ass:working} imposes nondegeneracy of the working conditional imputation models and yields the uniform minorisation condition required for geometric ergodicity of the inner Metropolis--Hastings kernel in Theorem~\ref{thm:ergodicity}, following the substantive-model-compatible FCS framework of \citet{bartlett2015multiple}. Assumption~\ref{ass:missing} formalizes the MAR mechanism and excludes degenerate missingness regimes as the number of missing covariates diverges with $n$. Assumptions~\ref{ass:ridge} and~\ref{ass:sparsity} specify the sparsity and regularization regime required for high-dimensional debiasing and asymptotic normality under diverging dimension \citep{xia2022statistical}. In particular, these assumptions control the stochastic error of the ridge-based working models, the Cox lasso estimation error, and the nodewise inverse-information approximation as $p_n\to\infty$.
Assumptions~\ref{ass:Sin} and~\ref{ass:Lout} are algorithmic conditions ensuring sufficient inner-chain mixing and convergence of the outer IRO iterates under increasing missing-data dimension \citep{wang1998mi,hughes2014joint}. Assumption~\ref{ass:contract} is a local stochastic stability condition on the coupled imputation-and-optimization map that controls the asymptotic behavior of the outer IRO iterates. The assumption is supported in practice by the convexity of the Cox lasso update, ridge-stabilized working models, bounded linear predictors, and empirical stabilization of successive iterates, although formal verification of the full coupled map is generally difficult. Finally, Assumption~\ref{ass:burnin} ensures that the retained completed datasets used for the Rubin pooling are approximately sampled from the stationary imputation distribution associated with the limiting substantive Cox model, with the burn-in length increasing appropriately with the diverging missing-data dimension.

 \begin{theorem}[Geometric ergodicity of the inner SMC-FCS chain]
\label{thm:ergodicity}
 Fix $\beta$ satisfying $\|\beta-\beta^0\|_1\le r$ for some fixed $r>0$, and let $\widehat\Lambda_0$ denote the Breslow baseline cumulative hazard estimator computed from the current completed data and the current value of $\beta$. Let $d_n=|\{(i,j):R_{ij}=0\}|$ denote the total number of missing covariate entries. Under Assumptions~\ref{ass:dimension}--\ref{ass:working} and~\ref{ass:missing}, with probability tending to one as $n\to\infty$, the inner SMC-FCS sweep kernel $P_{\beta,\widehat\Lambda_0}$ admits a unique stationary distribution $\nu_{\beta,\widehat\Lambda_0}$. Moreover, there exist constants $C_0<\infty$ and $\varepsilon_a\in(0,1)$, independent of $(n,p_n)$, such that $\rho_{d_n}=1-\varepsilon_a^{d_n}\in(0,1)$ and
\[
\left\|
P_{\beta,\widehat\Lambda_0}^S(x_0,\cdot)
-
\nu_{\beta,\widehat\Lambda_0}
\right\|_{\mathrm{TV}}
\le
C_0
\rho_{d_n}^S
\]
for every initial state $x_0\in\mathcal X$ and every $S\ge1$.
Consequently, for any target total-variation accuracy $\delta>0$, it suffices that $S\ge \log(C_0/\delta)/\varepsilon_a^{d_n}$. In particular, because $d_n\le \bar r\,n p_n$ under Assumption~\ref{ass:missing}, the mixing rate deteriorates with the diverging missing-data dimension, requiring the number of inner sweeps to increase accordingly with $(n,p_n)$.
\end{theorem}

 The results imply that the inner SMC-FCS chain converges geometrically fast to a unique stationary distribution conditional on the current substantive-model parameters. Consequently, after a logarithmic number of inner sweeps, the completed datasets used in the outer IRO updates become asymptotically insensitive to initialization, providing the key mixing condition needed for the convergence of the outer iterates established in the next theorem.

 \begin{theorem}[Convergence of the IRO iterates]
\label{thm:rate}
Let $(\widetilde\beta^{(\ell)},\widehat\Lambda_0^{(\ell)})$ denote the outer IRO iterates, where $\widetilde\beta^{(\ell)}$ is the Cox lasso estimator computed from the completed data after the $\ell$th imputation update and $\widehat\Lambda_0^{(\ell)}$ is the corresponding Breslow baseline cumulative hazard estimator. Under Assumptions~\ref{ass:dimension}--\ref{ass:missing}, \ref{ass:Sin}, \ref{ass:Lout}, and~\ref{ass:contract}, suppose $\sup_{n,\ell}\E\|\widetilde\beta^{(\ell)}\|_1<\infty$. Then, as $n\to\infty$ and $\ell\to\infty$ satisfying Assumption~\ref{ass:Lout},
\[
\|\widetilde\beta^{(\ell)}-\beta^\ast\|_1
+
\|\widehat\Lambda_0^{(\ell)}-\Lambda_0^\ast\|_\infty
=
o_p(1),
\]
for almost every observed-data sequence, where $(\beta^\ast,\Lambda_0^\ast)$ is the fixed point of the population IRO map induced by the stationary SMC-FCS law together with the Cox lasso and Breslow updates. Moreover, the finite-sweep approximation error from the inner SMC-FCS chain is asymptotically negligible under Assumption~\ref{ass:Sin}, despite the diverging missing-data dimension $d_n$.
\end{theorem}
 
 Theorem~\ref{thm:rate} shows that the coupled imputation-and-optimization procedure converges to a stable fixed point of the population IRO map for both the regression coefficient and the baseline cumulative hazard. Combined with Theorem~\ref{thm:ergodicity}, this implies that the finite-sweep SMC-FCS updates introduce asymptotically negligible error, so the completed datasets behave asymptotically as draws from a stationary imputation mechanism compatible with the substantive Cox model. This stationary regime provides the foundation for the high-dimensional estimation and inference results developed in Theorems~\ref{thm:lassorate} and~\ref{thm:clt}.

 \begin{theorem}[Cox lasso rate under IRO stationarity]
\label{thm:lassorate}
Suppose the conditions of Theorem~\ref{thm:rate} hold together with Assumptions~\ref{ass:eigen}, \ref{ass:ridge}, and \ref{ass:sparsity}. Assume the outer IRO iterates have reached stationarity at the fixed point $(\beta^\ast,\Lambda_0^\ast)$ under the triangular-array regime $p=p_n\to\infty$. Let $(\widetilde\beta^{(\ell)},\widehat\Lambda_0^{(\ell)})$ denote the Cox lasso estimator and Breslow baseline cumulative hazard estimator obtained at the $\ell$th outer IRO iteration from the corresponding completed dataset. Then, as $n\to\infty$,
\[
\|\widetilde\beta^{(\ell)}-\beta^\ast\|_1
=
O_p\!\left(
s_0\sqrt{\frac{\log p_n}{n}}
\right).
\]

If additionally
\(
\|\beta^\ast-\beta^0\|_1
=
O\!\left(
s_0\sqrt{\frac{\log p_n}{n}}
\right),
\)
then it holds that
\[
\|\widetilde\beta^{(\ell)}-\beta^0\|_1
=
O_p\!\left(
s_0\sqrt{\frac{\log p_n}{n}}
\right).
\]

Moreover,
\[
\|\widehat\Lambda_0^{(\ell)}-\Lambda_0^\ast\|_\infty
=
o_p(1).
\]

The above rates hold uniformly over the diverging-dimensional sequence $p_n=O(n^\kappa)$ and account for both the stationary imputation variability and the finite-sweep approximation error from the inner SMC-FCS chain.
\end{theorem}

  The results suggest that, after the IRO algorithm reaches stationarity, the Cox lasso estimator retains the standard high-dimensional $\ell_1$ convergence rate despite missing covariates and iterative imputation updates, while the associated Breslow baseline cumulative hazard estimator remains uniformly consistent. These rate results lead to the asymptotic normality of the debiased Rubin-pooled estimator in Theorem~\ref{thm:clt}.

 \begin{theorem}[Asymptotic normality for linear contrasts]
\label{thm:clt}
Suppose Assumptions~\ref{ass:dimension}--\ref{ass:burnin} hold. Assume $\lambda_n\asymp\sqrt{\log p_n/n}$ and $\gamma_n\asymp\|\Theta_{\beta^0}\|_{1,1}s_0\lambda_n$, and suppose the conclusions of Theorem~\ref{thm:lassorate} hold. Assume the SMC-FCS imputation procedure is asymptotically proper and congenial with the Cox substantive model in the sense of \cite{wang1998mi}.
Let $\bar\beta_M=M^{-1}\sum_{m=1}^M\widehat\beta_{\mathrm{db}}^{(m)}$ denote the Rubin-pooled debiased estimator. Then, for any nonrandom loading vector $c_n\in\mathbb R^{p_n}$ satisfying $\|c_n\|_2=1$ and $\|c_n\|_1\le a^\ast<\infty$,
 \[
\frac{\sqrt n\,
c_n^\top
(\bar\beta_M-\beta^0)} {
\sqrt{c_n^\top
\left[
V_{\mathrm{com}}
+
\left(
1+\frac1M
\right)
V_{\mathrm{mis}}
\right]
c_n}}
\overset{d}{\longrightarrow}
\mathcal N(0,1),
\]
 where $V_{\mathrm{com}}=\Theta_{\beta^0}$ and $V_{\mathrm{mis}}=V_{\mathrm{ML}}-V_{\mathrm{com}}$.
 Moreover, 
\[
n\,c_n^\top\widehat V_{\mathrm{total}}c_n  \Big/ c_n^\top
\left[
V_{\mathrm{com}}
+
\left(1+\frac1M\right)
V_{\mathrm{mis}}
\right]
c_n
\overset{p}{\longrightarrow}
 1,
\]
and
\[\frac{
\sqrt n\,
c_n^\top
(\bar\beta_M-\beta^0)
}{
\sqrt{
n\,
c_n^\top
\widehat V_{\mathrm{total}}
c_n
}
}
\overset{d}{\longrightarrow}
\mathcal N(0,1).
\]
Thus, Rubin-based Studentization is asymptotically valid for  fixed $M$ under the diverging-dimensional regime.
\end{theorem}

Theorem~\ref{thm:clt} establishes asymptotically valid inference for linear  contrasts after debiasing and Rubin pooling across imputations. The result shows that the proposed SMC-DBL procedure correctly propagates uncertainty from missing covariates and attains the semiparametric observed-data efficiency bound when the number of imputations diverges, while remaining asymptotically valid for fixed finite $M$.

\section{Simulations}  \label{sec:simulation}
 
We investigated the operating characteristics of the proposed SMC-DBL procedure through simulation. Four competitors are considered, including an oracle debiased lasso fitted to the fully observed data, the IRO procedure of \citet{liang2018imputation} without debiasing, the standard SMC-FCS imputation of \citet{bartlett2015multiple} followed by per-imputation debiased lasso and Rubin pooling, and the mean-or-mode imputation followed by debiased lasso. The oracle is treated as a gold standard in the sense that it has access to data the other procedures must reconstruct from observed entries.  
 
\vspace{1ex}
 
We considered all combinations of $n \in \{500, 1000, 2000\}$.  For notational convenience, we suppress the dependence of \(p\) on \(n\) and consider \(p\in\{20,50,100,200\}\), yielding 12 simulation settings. The number of nonzero coefficients was fixed at $s_0 = 5$, with
\begin{equation*}
\beta^0 \;=\; (1.0,\, 1.0,\, 1.0,\, 0.5,\, 0.5,\, 0,\, \ldots,\, 0)\T.
\end{equation*}
Covariates were drawn from a multivariate Gaussian distribution with mean zero and AR(1) covariance $\Sigma_{ij} = \rho^{|i-j|}$, $\rho = 0.5$, then truncated coordinatewise to $[-5, 5]$ to satisfy the boundedness assumption underlying the Cox-model partial-likelihood theory. Failure times were drawn from an exponential distribution with rate $\exp(X\T \beta^0)$, and censoring times were drawn independently from an exponential distribution with rate 0.1. The marginal censoring fraction ranged from approximately 30\% to 45\% across settings.
 
\vspace{1ex}
 
Missingness was placed on coordinates $2,\dots,\lfloor 0.2 p \rfloor + 1$, so that 20\% of covariates were partially missing in each setting. The missingness indicator at coordinate $j$ followed a Bernoulli draw with probability $\mathrm{logit}^{-1}(-1.5 + 0.8 X_1)$ clipped to $[0.80, 0.95]$. The resulting mechanism is missing at random, with the fully observed first coordinate as the sole driver, matching the conditioning used in the theoretical development. The realised per-coordinate missing rate was approximately 5--20\% on affected columns, with a subject-level missing rate of roughly 50\% under union over coordinates and an overall cell-level rate of approximately 3--4\% on the design matrix.
 
\vspace{1ex}
 
For each $(n, p)$ pair, we generated $R = 100$ independent datasets and applied all five procedures. The SMC-DBL procedure used $M = 20$ inferential chains, each running an independent Phase~1 IRO convergence followed by a per-chain burn-in of $T_{\mathrm{burn}} = \max\{20, \lceil 3 \log(np)\rceil\}$ iterations, with $S_{\mathrm{inner}} = \max\{5, \lceil \log(np)\rceil\}$ inner sweeps per outer iteration. The ridge multiplier was tuned by five-fold cross-validation on the prediction error of the first variable in $\mcM$, with the response taken from rows on which that variable is genuinely observed and the predictor matrix read from the same rows of the imputed pseudo-complete dataset produced by a single tuning Phase~1 chain run at a small nominal ridge $b_0 = 0.1$. The minimum-MSE rule was used over a candidate grid for $b$ ranging from $0.01$ to $\max(2, 20\sqrt{p/n})$ on a 15-point log-spaced grid. The QP multiplier in the debiasing step was fixed at $a = 0.5$. The converged Phase~1 state of the tuning chain was discarded, and the $M$ inferential chains shared no Phase~1 endpoint and were initialised independently.

Three operating-characteristic figures summarise the comparison on the active set $\mathcal{S} = \{k : \beta^0_k \neq 0\}$. Per-coordinate detail across all settings is reported in Appendix Table \ref{tab:sim_nonzero_full}. The active set is the inferential target of practical interest, and reporting averages over $\mathcal{S}$ yields a compact summary that nevertheless preserves the essential contrast among methods. Mean absolute bias on $\mathcal{S}$ is $\overline{|\mathrm{Bias}|}_{\mathcal{S}} = |\mathcal{S}|^{-1} \sum_{k \in \mathcal{S}} |\mathrm{Bias}_k|$, with $\mathrm{Bias}_k = R^{-1} \sum_{r=1}^R (\hat{\beta}_k^{(r)} - \beta_k^0)$, and the active-set averages of empirical SD, mean estimated SE, and coverage are defined analogously.

\begin{figure}[t]
  \centering
  \includegraphics[width=\linewidth]{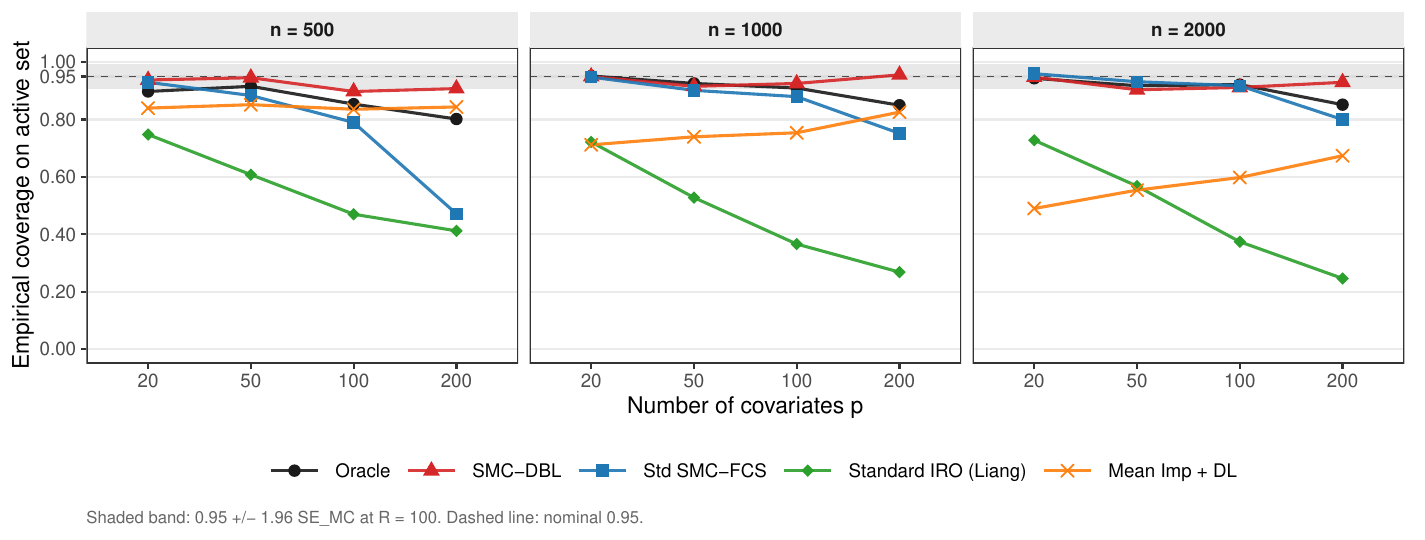}
  \caption{Empirical coverage on the active set across the twelve
    simulation settings.}
  \label{fig:coverage}
\end{figure}

\vspace{1ex}

SMC-DBL holds the nominal $0.95$ level within Monte-Carlo error in eleven of the twelve settings, Std~SMC-FCS achieves nominal coverage at $n = 1000$ and $n = 2000$ for $p \leq 100$ but drops to $0.47$ at $(n=500, p=200)$ and to $0.75$ at $(n=1000, p=200)$, reflecting the breakdown of the univariate normal imputer under  working models. The oracle, despite having access to the fully observed data, exhibits mild under-coverage (roughly $0.80$--$0.85$) at $p = 200$ for all $n$, a feature attributable to the diverging-$p$ remainder in the debiased-lasso decomposition rather than to the imputation step. Standard IRO under-covers severely and the gap widens with $p$: from $0.75$ at
$(n=500, p=20)$ down to $0.25$--$0.41$ at $p=200$, reflecting the absence of a debiasing correction. Mean Imp~+~DL produces low coverage that worsens with $n$ at small $p$, illustrating that single-imputation procedures fail to reflect imputation uncertainty: as $n$ grows the standard errors shrink at the $\sqrt{n}$ rate but the imputation bias does not, so the coverage gap widens.
 
\vspace{1ex}

\begin{figure}[t]
  \centering
  \includegraphics[width=\linewidth]{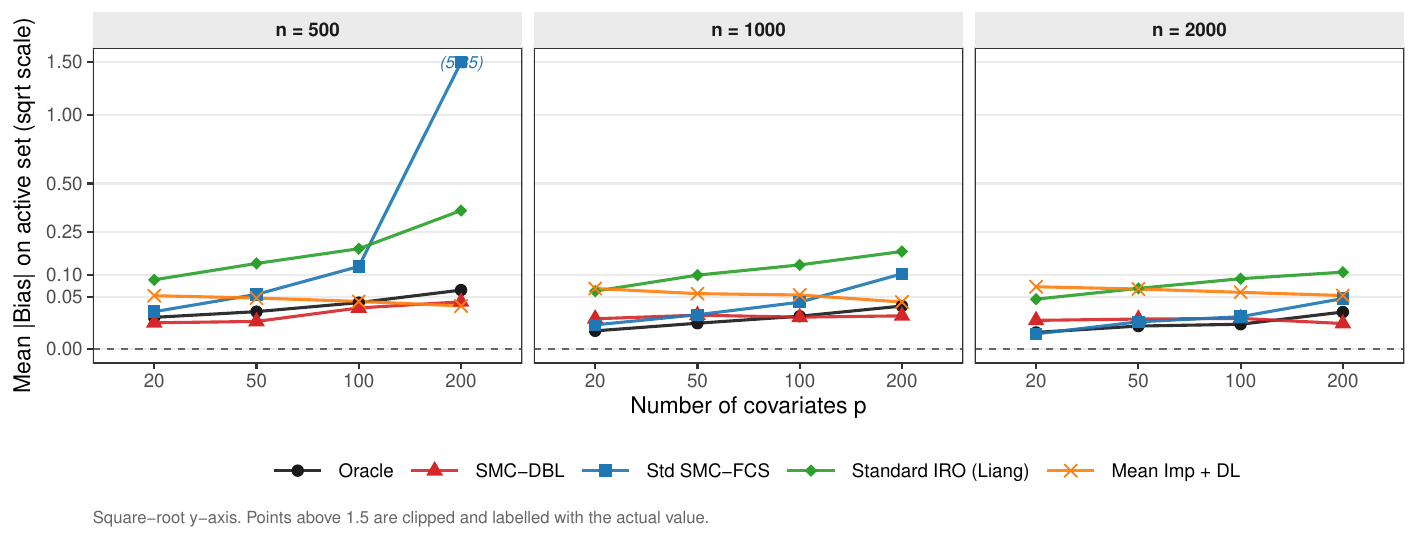}
  \caption{Mean absolute bias on the active set.}
  \label{fig:bias}
\end{figure}

On the square-root scale, SMC-DBL tracks the oracle closely across all twelve settings, with the active-set absolute bias in the range $0.02$ to $0.05$. The two methods are visually indistinguishable at $n \in \{1000, 2000\}$. Std~SMC-FCS tracks SMC-DBL at $p \leq 100$ but exhibits a catastrophic failure at $(n=500, p=200)$ where the active-set bias reaches approximately $5.5$. The failure reflects the breakdown of unregularised univariate
imputation in the regime $p \approx n/2$ with $20\%$ of columns missing. Standard IRO carries a non-trivial active-set bias at every setting, in the
range $0.05$ to $0.35$, increasing with $p$ and decreasing with $n$ as expected for a Cox-lasso point estimate without debiasing. Mean~Imp~+~DL has
bias comparable to the oracle for small $p$ but drifts upward as $p$ grows, again reflecting the loss of correlation information when missing entries are replaced by column means.

\vspace{1ex}

\begin{figure}[t]
  \centering
  \includegraphics[width=\linewidth]{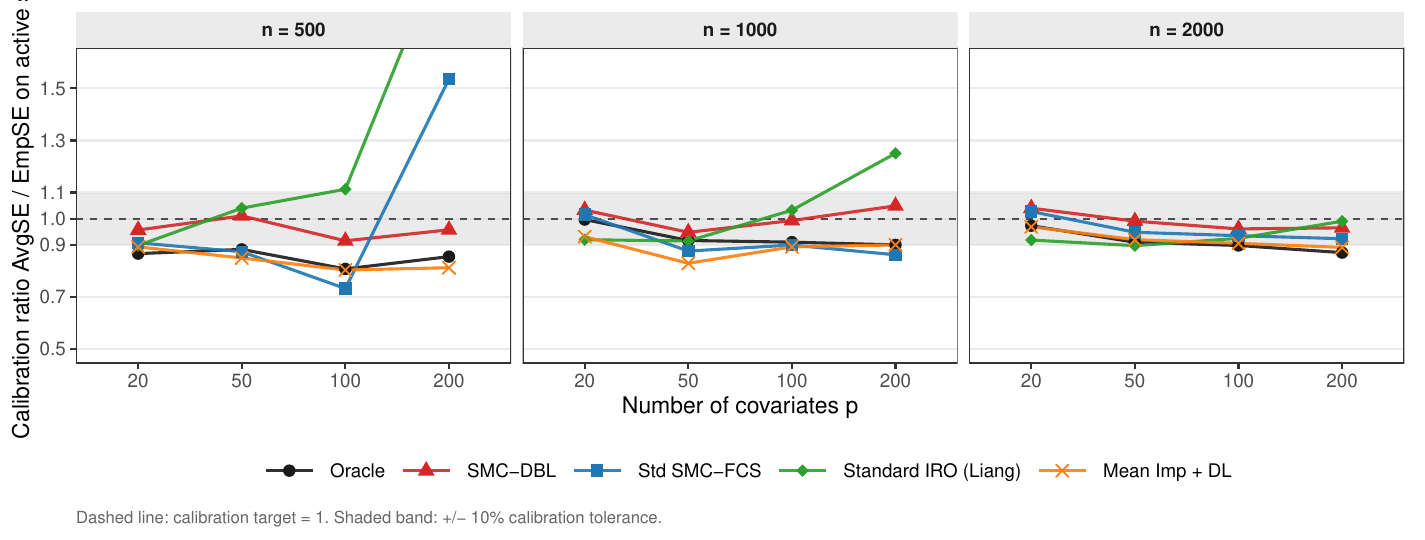}
  \caption{Variance calibration on the active set}
  \label{fig:calibration}
\end{figure}

The calibration ratio $\widehat{\mathrm{AvgSE}} / \widehat{\mathrm{EmpSE}}$ takes the value one when the model-based standard errors target the empirical sampling SD; values below one indicate under-estimation of sampling variability, values above one indicate conservatism. SMC-DBL falls inside the $\pm 10\%$ tolerance band in eleven of the twelve settings, with the largest deviation $0.92$ occurring at $(n=500, p=100)$. The ratio rises to $1.05$ at $(n=1000, p=200)$, in line with the proper-imputation limit $V_{\mathrm{ML}} + V_{\mathrm{mis}}/M$ from Theorem~\ref{thm:clt}, which inflates the within-imputation sandwich by a factor that grows with the missing-information fraction. Mean~Imp~+~DL clusters near $0.85$--$0.90$ for $n \leq 1000$ but moves into the calibration band at $n = 2000$, showing that single-imputation under-estimation persists at moderate $n$ but attenuates as the sampling variance dominates the imputation variance. Std~SMC-FCS is well calibrated at $p \leq 50$ but jumps to $1.55$ at $(n=500, p=200)$, a mechanical consequence of the catastrophic bias point in Figure~\ref{fig:bias}: the empirical SD on the active set is dominated by the bias spike, but the model-based sandwich does not register that spike, so the ratio is in fact misleadingly close to one only because both numerator and denominator are large rather than because the procedure is calibrated. Standard IRO produces ratios above the upper tolerance at $p \geq 100$ for $n \in \{500, 1000\}$, peaking near $1.7$ at $(n=500, p=200)$; the plug-in Hessian-inverse used in the IRO variance estimate is over-estimating the sampling SD in regimes where shrinkage bias dominates the Monte-Carlo fluctuation.

\section{Data Analysis} \label{sec:data}
\subsection{Data}
 
The application uses a non-small cell lung cancer sub-cohort comprising 979 European-ancestry patients with pathologically confirmed disease, available germline genotype data, and complete survival information. 5 patients with missing outcomes were excluded from the analysis. Genotyping was performed on Illumina arrays. Standard sample- and variant-level quality control was applied within the cohort, including filters on sample missingness, variant call rate, minor allele frequency, and Hardy-Weinberg equilibrium. Cryptic relatedness was addressed during upstream genotype quality control. Three within-cohort ancestry principal components, computed by principal component analysis on linkage-disequilibrium-pruned genome-wide markers (PLINK v1.9, --indep-pairwise 50 5 0.2), were retained for adjustment. Overall survival was defined as the time from pathological diagnosis to death from any cause, with patients still alive at last follow-up censored at that date.

 \vspace{1ex}
 
Seven clinical covariates were extracted: age at diagnosis (continuous), sex (binary), smoking status (never, former, or current), disease stage (early stages I--II vs late stages III--IV), surgical resection (yes or no), chemotherapy (yes or no), radiation therapy (yes or no), and the first three within-cohort ancestry principal components. The treatment indicators record whether each treatment was ever received, recorded at the time of data abstraction. Because treatment timing was not uniformly available, the indicators may be subject to immortal-time bias, and we therefore treat their coefficients as adjustment quantities rather than causal effect estimates. 

 \vspace{1ex}

The genetic predictors were drawn from a prespecified panel of 77 candidate single-nucleotide polymorphisms assembled from lung-function GWAS evidence, primarily from the \cite{kachuri2020lung} instrument set, supplemented with representative lead variants from established pleiotropic loci. Variants were located by chromosomal position under the GRCh37/hg19 reference build. After quality control, 51 of the 77 candidate variants were available in the cohort and entered the analysis. The remaining variants were either monomorphic, failed Hardy-Weinberg or call-rate filters, or were not present on the genotyping array. The detailed descriptive table is reported in Appendix Table \ref{tab:table1}.

\subsection{Comparisons of Results}
\label{sec:app_settings}

 We applied the SMC-DBL procedure described in Section~\ref{sec:method}, using a design vector that included both clinical covariates and HSPH SNP genotypes. Clinical covariates were standardised to unit variance prior to penalisation, whereas SNPs were retained on the additive \(0/1/2\) dosage scale. The resulting design matrix contained \(p=61\) covariates.

\vspace{1ex}

Table~\ref{tab:realdata} compares SMC-DBL with several competing approaches considered in the simulation studies and further illustrates the advantages of the proposed method. Compared with complete-case (CC) analysis, SMC-DBL generally produced smaller standard errors and stronger statistical significance for both clinical covariates and SNP effects, reflecting improved efficiency through the use of partially observed subjects rather than discarding incomplete cases. For example, the estimated effect of late-stage disease increased from \(0.772\) (SE \(=0.257\)) under CC analysis to \(0.929\) (SE \(=0.173\)) under SMC-DBL. SMC-DBL and standard SMC-FCS yielded highly similar estimates, suggesting that the proposed debiased lasso method preserves the substantive-model-compatible imputation structure while extending inference to diverging-dimensional settings. In contrast, the IRO approach substantially attenuated most SNP effects toward zero, frequently producing \(p\)-values close to one. For instance, rs11022690 was statistically significant under SMC-DBL (\(\hat\beta=-0.113\), \(p=0.022\)) but was considerably attenuated under IRO (\(\hat\beta=-0.054\), \(p=0.300\)). Similar patterns were observed for rs72743477, indicating that SMC-DBL better preserves moderate prognostic signals in the presence of missingness. Although mean imputation produced estimates numerically similar to SMC-DBL for some variables, it does not properly account for uncertainty due to missing data and therefore lacks formal inferential validity. In contrast, SMC-DBL combines compatible multiple imputation, Rubin's variance decomposition, and debiased penalized estimation in a unified fashion, providing principled inference for moderately high-dimensional Cox regression with missing covariates.

\vspace{1ex}

 Among the identified variants, rs72743477 is located within \textit{SMAD3}, a central mediator of the TGF-\(\beta\) signaling pathway that regulates cell proliferation, differentiation, immune modulation, and extracellular matrix remodeling. Aberrant TGF-\(\beta\)/\textit{SMAD3} signaling has been widely implicated in lung cancer progression, particularly in epithelial--mesenchymal transition (EMT), tumor invasion, metastatic dissemination, and resistance to therapy. In NSCLC, increased \textit{SMAD3} activity has been associated with aggressive tumor phenotypes, immune evasion, fibrosis-related stromal remodeling, and poorer prognosis. Experimental studies have also suggested that \textit{SMAD3} contributes to tumor-promoting inflammatory responses and may interact with smoking-related oxidative stress pathways, both of which are relevant to lung carcinogenesis. These biological roles support the plausibility that variation within the \textit{SMAD3} region could influence survival outcomes through effects on tumor progression and the tumor microenvironment. In contrast, rs11022690 currently has limited functional annotation and has not been extensively characterized in lung cancer studies. Nevertheless, the consistency of its estimated protective effects across multiple imputation-based analyses suggests a potentially stable prognostic signal rather than a spurious association driven by missing-data handling. It is possible that rs11022690 tags a nearby regulatory region or acts through linkage disequilibrium with functional variants affecting gene expression or immune-related pathways. Additional analyses integrating expression quantitative trait locus (eQTL) data, epigenomic annotations, and transcriptomic profiling may help clarify its biological relevance.

\begin{landscape}
\begingroup\footnotesize
\begin{longtable}{lrlrlrlrlrl}
\caption{Pooled estimates for all predictors across the five comparison procedures applied to the NSCLC cohort. }\label{tab:realdata}\\
\toprule
Variable & \multicolumn{2}{c}{CC} & \multicolumn{2}{c}{SMC-DBL} & \multicolumn{2}{c}{Std SMC-FCS} & \multicolumn{2}{c}{IRO} & \multicolumn{2}{c}{Mean Imp}\\
\cmidrule(lr){2-3} \cmidrule(lr){4-5} \cmidrule(lr){6-7} \cmidrule(lr){8-9} \cmidrule(lr){10-11}
 & $\hat{\beta}$ (SE) & P-value & $\hat{\beta}$ (SE) & P-value & $\hat{\beta}$ (SE) & P-value & $\hat{\beta}$ (SE) & P-value & $\hat{\beta}$ (SE) & P-value\\
\midrule
\endfirsthead
\multicolumn{11}{l}{\textit{Table~\ref{tab:realdata} continued}}\\
\toprule
Variable & \multicolumn{2}{c}{CC} & \multicolumn{2}{c}{SMC-DBL} & \multicolumn{2}{c}{Std SMC-FCS} & \multicolumn{2}{c}{IRO} & \multicolumn{2}{c}{Mean Imp}\\
\cmidrule(lr){2-3} \cmidrule(lr){4-5} \cmidrule(lr){6-7} \cmidrule(lr){8-9} \cmidrule(lr){10-11}
 & $\hat{\beta}$ (SE) & $p$ & $\hat{\beta}$ (SE) & $p$ & $\hat{\beta}$ (SE) & $p$ & $\hat{\beta}$ (SE) & $p$ & $\hat{\beta}$ (SE) & $p$\\
\midrule
\endhead
\midrule
\multicolumn{11}{r}{\textit{Continued on next page}}\\
\endfoot
\bottomrule
\endlastfoot
\multicolumn{11}{l}{\textit{Clinical covariates}}\\
Surgery (yes vs no) & -1.252 (0.179) & $<$0.001 & -1.117 (0.123) & $<$0.001 & -1.100 (0.124) & $<$0.001 & -1.000 (0.116) & $<$0.001 & -1.115 (0.123) & $<$0.001\\
Late stage (III-IV vs I-II) & 0.772 (0.257) & 0.003 & 0.929 (0.173) & $<$0.001 & 0.955 (0.176) & $<$0.001 & 0.586 (0.163) & $<$0.001 & 0.930 (0.173) & $<$0.001\\
Age (per SD) & 0.222 (0.054) & $<$0.001 & 0.259 (0.038) & $<$0.001 & 0.257 (0.038) & $<$0.001 & 0.194 (0.041) & $<$0.001 & 0.258 (0.038) & $<$0.001\\
Sex (male vs female) & -0.280 (0.104) & 0.007 & -0.328 (0.071) & $<$0.001 & -0.328 (0.071) & $<$0.001 & -0.200 (0.074) & 0.007 & -0.326 (0.071) & $<$0.001\\
Smoking (per level) & 0.205 (0.083) & 0.014 & 0.232 (0.059) & $<$0.001 & 0.232 (0.059) & $<$0.001 & 0.122 (0.061) & 0.047 & 0.231 (0.059) & $<$0.001\\
Radiation (yes vs no) & -0.093 (0.128) & 0.471 & -0.172 (0.088) & 0.051 & -0.175 (0.088) & 0.048 & 0.000 (0.100) & 1.000 & -0.175 (0.088) & 0.048\\
Chemotherapy (yes vs no) & -0.040 (0.217) & 0.854 & -0.025 (0.141) & 0.859 & -0.041 (0.142) & 0.775 & 0.000 (0.143) & 1.000 & -0.031 (0.141) & 0.829\\
PC1 (per SD) & 0.057 (0.057) & 0.320 & 0.058 (0.036) & 0.107 & 0.058 (0.036) & 0.111 & 0.022 (0.042) & 0.596 & 0.059 (0.036) & 0.104\\
PC2 (per SD) & -0.032 (0.059) & 0.594 & -0.005 (0.035) & 0.887 & -0.004 (0.035) & 0.906 & 0.000 (0.038) & 1.000 & -0.005 (0.035) & 0.897\\
PC3 (per SD) & 0.035 (0.050) & 0.484 & -0.016 (0.034) & 0.645 & -0.016 (0.034) & 0.639 & 0.000 (0.040) & 1.000 & -0.016 (0.034) & 0.651\\
\midrule
\multicolumn{11}{l}{\textit{SNPs (ordered by SMC-DBL $p$-value)}}\\
rs11022690 (C) & -0.129 (0.070) & 0.066 & -0.113 (0.049) & 0.022 & -0.112 (0.049) & 0.024 & -0.054 (0.052) & 0.300 & -0.114 (0.049) & 0.022\\
rs72743477 (G) & -0.196 (0.090) & 0.029 & -0.126 (0.061) & 0.038 & -0.127 (0.061) & 0.037 & -0.021 (0.062) & 0.732 & -0.126 (0.061) & 0.038\\
rs72490631 (C) & 0.084 (0.084) & 0.316 & -0.107 (0.060) & 0.074 & -0.108 (0.060) & 0.072 & -0.017 (0.063) & 0.789 & -0.108 (0.060) & 0.071\\
rs196025 (A) & -0.120 (0.074) & 0.104 & -0.086 (0.052) & 0.099 & -0.089 (0.053) & 0.091 & -0.028 (0.055) & 0.602 & -0.090 (0.052) & 0.083\\
rs4233430 (T) & -0.082 (0.099) & 0.409 & -0.108 (0.067) & 0.105 & -0.109 (0.067) & 0.105 & -0.002 (0.068) & 0.982 & -0.112 (0.067) & 0.095\\
rs17032590 (G) & -0.096 (0.090) & 0.290 & -0.102 (0.063) & 0.107 & -0.099 (0.063) & 0.117 & -0.009 (0.066) & 0.890 & -0.105 (0.064) & 0.100\\
rs77972916 (A) & -0.290 (0.163) & 0.075 & -0.180 (0.117) & 0.126 & -0.182 (0.117) & 0.121 & 0.000 (0.110) & 1.000 & -0.183 (0.117) & 0.118\\
rs9660890 (C) & -0.086 (0.089) & 0.338 & -0.089 (0.061) & 0.144 & -0.089 (0.061) & 0.146 & -0.003 (0.063) & 0.958 & -0.090 (0.061) & 0.141\\
rs7196853 (C) & 0.088 (0.134) & 0.512 & -0.124 (0.089) & 0.162 & -0.127 (0.088) & 0.147 & -0.000 (0.088) & 1.000 & -0.126 (0.088) & 0.152\\
rs72811372 (A) & -0.103 (0.145) & 0.478 & -0.129 (0.096) & 0.177 & -0.137 (0.096) & 0.153 & 0.000 (0.092) & 1.000 & -0.132 (0.095) & 0.164\\
rs11118683 (T) & 0.067 (0.074) & 0.363 & 0.067 (0.051) & 0.187 & 0.067 (0.051) & 0.186 & 0.002 (0.055) & 0.971 & 0.067 (0.051) & 0.185\\
rs17387279 (G) & 0.137 (0.090) & 0.128 & 0.077 (0.063) & 0.218 & 0.079 (0.063) & 0.209 & 0.000 (0.069) & 1.000 & 0.077 (0.063) & 0.222\\
rs28517513 (T) & 0.050 (0.090) & 0.582 & 0.067 (0.061) & 0.271 & 0.064 (0.061) & 0.295 & 0.000 (0.065) & 1.000 & 0.068 (0.061) & 0.260\\
rs1956028 (C) & 0.177 (0.112) & 0.116 & 0.081 (0.075) & 0.284 & 0.082 (0.075) & 0.278 & 0.000 (0.079) & 1.000 & 0.080 (0.075) & 0.290\\
rs11227223 (T) & 0.017 (0.163) & 0.917 & 0.119 (0.112) & 0.287 & 0.121 (0.112) & 0.279 & 0.000 (0.115) & 1.000 & 0.119 (0.112) & 0.289\\
rs4948502 (C) & 0.040 (0.073) & 0.586 & -0.049 (0.052) & 0.342 & -0.048 (0.052) & 0.348 & 0.000 (0.053) & 1.000 & -0.049 (0.052) & 0.342\\
rs10878300 (T) & 0.078 (0.091) & 0.396 & 0.060 (0.066) & 0.362 & 0.059 (0.066) & 0.374 & 0.000 (0.070) & 1.000 & 0.061 (0.066) & 0.354\\
rs4233284 (G) & 0.040 (0.078) & 0.610 & 0.043 (0.052) & 0.409 & 0.041 (0.052) & 0.433 & 0.000 (0.056) & 1.000 & 0.044 (0.052) & 0.401\\
rs10173269 (G) & -0.060 (0.077) & 0.431 & -0.042 (0.051) & 0.411 & -0.046 (0.051) & 0.368 & -0.000 (0.053) & 0.999 & -0.044 (0.052) & 0.393\\
rs12466981 (T) & 0.016 (0.081) & 0.847 & -0.042 (0.055) & 0.439 & -0.044 (0.055) & 0.426 & 0.000 (0.059) & 1.000 & -0.043 (0.055) & 0.428\\
rs10184235 (G) & 0.158 (0.089) & 0.075 & 0.044 (0.057) & 0.443 & 0.046 (0.057) & 0.421 & 0.000 (0.061) & 1.000 & 0.044 (0.058) & 0.444\\
rs1528624 (G) & -0.024 (0.071) & 0.739 & 0.035 (0.049) & 0.467 & 0.039 (0.049) & 0.425 & 0.000 (0.052) & 1.000 & 0.036 (0.049) & 0.461\\
rs4800410 (C) & 0.076 (0.078) & 0.326 & 0.034 (0.053) & 0.527 & 0.038 (0.055) & 0.482 & 0.000 (0.056) & 1.000 & 0.035 (0.053) & 0.511\\
rs6006399 (G) & -0.104 (0.137) & 0.450 & 0.051 (0.085) & 0.547 & 0.050 (0.085) & 0.552 & 0.000 (0.088) & 1.000 & 0.052 (0.085) & 0.543\\
rs7443323 (A) & -0.007 (0.080) & 0.935 & 0.034 (0.056) & 0.549 & 0.031 (0.056) & 0.580 & 0.000 (0.061) & 1.000 & 0.034 (0.056) & 0.551\\
rs4886509 (C) & -0.015 (0.080) & 0.849 & 0.031 (0.054) & 0.562 & 0.034 (0.054) & 0.531 & 0.000 (0.056) & 1.000 & 0.035 (0.054) & 0.521\\
rs17821105 (T) & -0.035 (0.101) & 0.732 & -0.039 (0.069) & 0.575 & -0.040 (0.069) & 0.560 & 0.000 (0.071) & 1.000 & -0.039 (0.070) & 0.577\\
rs11745375 (T) & 0.125 (0.073) & 0.088 & 0.023 (0.050) & 0.641 & 0.024 (0.050) & 0.627 & 0.000 (0.054) & 1.000 & 0.023 (0.050) & 0.645\\
rs4889526 (A) & -0.144 (0.080) & 0.073 & -0.024 (0.053) & 0.643 & -0.026 (0.053) & 0.621 & 0.000 (0.056) & 1.000 & -0.025 (0.053) & 0.634\\
rs9819463 (C) & -0.093 (0.091) & 0.306 & 0.028 (0.062) & 0.654 & 0.029 (0.062) & 0.634 & 0.000 (0.068) & 1.000 & 0.029 (0.062) & 0.643\\
rs10987386 (T) & 0.076 (0.104) & 0.464 & 0.030 (0.067) & 0.655 & 0.030 (0.068) & 0.657 & 0.000 (0.071) & 1.000 & 0.030 (0.067) & 0.655\\
rs1548029 (C) & 0.024 (0.074) & 0.748 & 0.023 (0.051) & 0.657 & 0.023 (0.051) & 0.654 & 0.000 (0.055) & 1.000 & 0.023 (0.051) & 0.653\\
rs12313454 (G) & -0.007 (0.109) & 0.952 & -0.033 (0.076) & 0.670 & -0.033 (0.077) & 0.664 & 0.000 (0.080) & 1.000 & -0.033 (0.077) & 0.670\\
rs12571363 (A) & -0.154 (0.122) & 0.207 & -0.035 (0.083) & 0.671 & -0.035 (0.083) & 0.673 & 0.000 (0.088) & 1.000 & -0.036 (0.083) & 0.666\\
rs58453446 (C) & 0.013 (0.078) & 0.869 & 0.021 (0.053) & 0.696 & 0.020 (0.053) & 0.704 & 0.000 (0.057) & 1.000 & 0.021 (0.053) & 0.696\\
rs2794359 (A) & 0.168 (0.146) & 0.252 & 0.039 (0.100) & 0.698 & 0.046 (0.101) & 0.651 & 0.000 (0.104) & 1.000 & 0.045 (0.100) & 0.652\\
rs55993676 (T) & -0.041 (0.083) & 0.624 & 0.019 (0.056) & 0.738 & 0.017 (0.056) & 0.759 & 0.000 (0.062) & 1.000 & 0.018 (0.056) & 0.749\\
rs1561073 (T) & 0.040 (0.083) & 0.634 & 0.017 (0.056) & 0.753 & 0.019 (0.056) & 0.735 & 0.000 (0.061) & 1.000 & 0.020 (0.056) & 0.718\\
rs659398 (T) & 0.126 (0.084) & 0.134 & 0.018 (0.058) & 0.756 & 0.020 (0.058) & 0.724 & 0.000 (0.060) & 1.000 & 0.018 (0.058) & 0.756\\
rs7927422 (C) & -0.053 (0.071) & 0.450 & -0.015 (0.050) & 0.767 & -0.016 (0.050) & 0.757 & 0.000 (0.053) & 1.000 & -0.014 (0.050) & 0.781\\
rs4444235 (T) & -0.010 (0.073) & 0.896 & 0.014 (0.049) & 0.771 & 0.014 (0.049) & 0.769 & 0.000 (0.055) & 1.000 & 0.014 (0.049) & 0.771\\
rs28719767 (C) & -0.109 (0.080) & 0.172 & -0.015 (0.055) & 0.782 & -0.014 (0.055) & 0.795 & 0.000 (0.059) & 1.000 & -0.017 (0.056) & 0.759\\
rs9393688 (T) & -0.133 (0.083) & 0.111 & -0.015 (0.058) & 0.795 & -0.016 (0.058) & 0.778 & 0.000 (0.060) & 1.000 & -0.013 (0.059) & 0.827\\
rs1425794 (T) & 0.007 (0.072) & 0.920 & 0.012 (0.049) & 0.810 & 0.011 (0.049) & 0.819 & 0.000 (0.053) & 1.000 & 0.012 (0.049) & 0.810\\
rs2125126 (A) & 0.108 (0.102) & 0.291 & -0.012 (0.069) & 0.860 & -0.010 (0.069) & 0.881 & 0.000 (0.073) & 1.000 & -0.011 (0.070) & 0.869\\
rs11887136 (A) & 0.197 (0.113) & 0.081 & 0.014 (0.078) & 0.860 & 0.017 (0.078) & 0.826 & 0.000 (0.081) & 1.000 & 0.014 (0.078) & 0.860\\
rs207672 (T) & 0.064 (0.076) & 0.399 & 0.007 (0.053) & 0.901 & 0.004 (0.053) & 0.936 & 0.000 (0.055) & 1.000 & 0.007 (0.053) & 0.896\\
rs113638840 (G) & -0.003 (0.086) & 0.969 & -0.007 (0.061) & 0.911 & -0.005 (0.060) & 0.929 & 0.000 (0.064) & 1.000 & -0.007 (0.060) & 0.904\\
rs750739 (G) & 0.076 (0.105) & 0.470 & 0.008 (0.076) & 0.917 & 0.006 (0.076) & 0.933 & 0.000 (0.077) & 1.000 & 0.005 (0.076) & 0.947\\
rs4636990 (G) & -0.063 (0.070) & 0.372 & 0.004 (0.048) & 0.925 & 0.005 (0.048) & 0.913 & 0.000 (0.050) & 1.000 & 0.005 (0.048) & 0.921\\
rs1179500 (C) & -0.049 (0.083) & 0.554 & -0.000 (0.057) & 0.996 & -0.000 (0.057) & 0.993 & 0.000 (0.059) & 1.000 & 0.001 (0.057) & 0.984\\
\end{longtable}
\endgroup
\end{landscape}

\section{Discussion} \label{sec:disc}

This paper extends Cox regression inference to settings with a diverging number of covariates subject to missingness, addressing an increasingly common challenge in modern biomedical and genomic survival studies. The central methodological strategy is  modular; we combine Cox-compatible semiparametric multiple imputation through substantive-model-compatible fully conditional specification (SMC-FCS) with rejection sampling, followed by debiased lasso inference within each completed dataset and final combination through Rubin's rules. This separation of imputation and inference allows the method to preserve the semiparametric structure of the Cox model while maintaining valid statistical inference in moderately high-dimensional settings.

\vspace{1ex}

A key advantage of the proposed approach is that it avoids reliance on restrictive parametric models for the conditional distribution of missing covariates. Instead, compatibility with the Cox proportional hazards model is enforced through acceptance probabilities derived from the Cox likelihood contribution. This substantially reduces the risk of incompatibility-induced bias that may arise when standard fully conditional imputation models are used with semiparametric survival outcomes. The resulting procedure remains computationally feasible while accommodating censoring, complex covariate dependence structures, and moderate dimensionality.

\vspace{1ex}

The theoretical justification integrates several complementary components. First, the procedure builds on recent asymptotic theory for debiased lasso estimators in Cox regression under diverging dimensionality. Second, it leverages large-sample properties of compatible multiple imputation procedures for semiparametric survival models. Third, Rubin's variance decomposition provides a principled mechanism for propagating uncertainty arising from both imputation and penalized estimation. Together, these ingredients support asymptotically valid estimation and inference under regularity conditions when the number of covariates increases with sample size.

\vspace{1ex}

Several limitations and future research directions remain. First, the current method assumes that covariates are missing at random (MAR). Although MAR is standard in the multiple imputation literature, it may be unrealistic in many biomedical studies where missingness depends on unobserved patient characteristics, latent disease severity, or unmeasured clinical decisions. Future work should therefore investigate extensions under missing-not-at-random (MNAR) mechanisms. Possible directions include selection models, pattern-mixture formulations, shared-parameter models, or sensitivity-analysis  tailored to high-dimensional survival settings. Developing semiparametric compatibility conditions under MNAR mechanisms would be particularly important for ensuring robust inference in observational cancer studies and electronic health record data.

\vspace{1ex}

Second, the current theory primarily focuses on the diverging-dimensional regime where the number of covariates grows with sample size but remains smaller than the effective sample size. Extending the method to high or ultra-high-dimensional settings with $p \gg n$ represents an important next step. Such extensions would require additional methodological developments at both the imputation and inference stages. For example, imputing high-dimensional covariates may require sparse graphical models, low-rank latent factor structures, or regularized conditional generators to stabilize estimation. On the inference side, stronger sparsity assumptions, improved precision matrix estimation, or sample-splitting and cross-fitting strategies may be needed to maintain valid debiasing properties under extreme dimensionality. The interaction between penalized survival estimation and multiply imputed high-dimensional data also raises new theoretical questions regarding uniform convergence, error propagation across imputations, and post-selection validity.

\vspace{1ex}

Additional future directions include extending the method to time-dependent covariates, competing risks, recurrent events, and clustered or distributed survival data. Incorporating machine learning--based imputation engines while preserving substantive-model compatibility is another promising direction. More broadly, the proposed method provides a foundation for integrating principled missing-data handling with modern high-dimensional survival inference, with potential applications in genomics, radiomics, multimodal cancer studies, and large-scale electronic health record analyses.

\bibliographystyle{plainnat} 
\bibliography{references} 


 \section{Appendix: Proofs, Demographics of BLCS and Additional Simulation Results}
\label{app:first}

This section presents  technical proofs of Theorems 1--4, descriptive analysis of BLCS and additional numerical evidence supporting the main results of the paper.
 
\subsection{Proofs}
\label{app:proofs}

\setcounter{equation}{0}
\renewcommand{\theequation}{S.\arabic{equation}}

  \begin{proof}[Proof of Theorem~\ref{thm:ergodicity}]
Fix $\beta$ satisfying $\|\beta-\beta^0\|_1\le r$. Since the analysis is under the triangular-array regime $p=p_n=O(n^\kappa)$, all constants below are required to be uniform in $(n,p_n)$. By Assumptions~\ref{ass:bounded} and~\ref{ass:linearpred},
\begin{equation}
\label{eq:linear-predictor-bound}
|X_i^\top\beta|
\le
|X_i^\top\beta^0|
+
\|X_i\|_\infty\|\beta-\beta^0\|_1
\le
K_1+Kr
\equiv B
\end{equation}
uniformly in $i$ and $n$. Thus $e^{-B}\le \exp(X_i^\top\beta)\le e^B$ for all $i$.

Let $\widehat\Lambda_0$ denote the Breslow estimator based on the current completed data,
\[
\widehat\Lambda_0(t)
=
\int_0^t
\frac{dN(u)}
{\sum_{i=1}^n I(Y_i\ge u)\exp(X_i^\top\beta)},
\]
where $N(u)=\sum_{i=1}^n I(Y_i\le u,\Delta_i=1)$. By Assumption~\ref{ass:followup}, $\Pr(Y\ge\tau)=\pi_0>0$. Therefore, by the law of large numbers, with probability tending to one,
\begin{equation}
\label{eq:riskset-event}
n^{-1}\sum_{i=1}^n I(Y_i\ge\tau)
\ge
\pi_0/2 .
\end{equation}
This bound depends only on the risk indicators and is unaffected by the fact that $p_n$ diverges. On the event in~\eqref{eq:riskset-event}, for every $u\le\tau$, monotonicity of the risk sets and~\eqref{eq:linear-predictor-bound} imply
\begin{equation}
\label{eq:denominator-lower-bound}
\sum_{i=1}^n I(Y_i\ge u)\exp(X_i^\top\beta)
\ge
\sum_{i=1}^n I(Y_i\ge\tau)\exp(X_i^\top\beta)
\ge
n e^{-B}\pi_0/2 .
\end{equation}
Combining the Breslow representation and~\eqref{eq:denominator-lower-bound}, and using $N(\tau)\le n$, gives
\begin{equation}
\label{eq:breslow-bound}
0\le \widehat\Lambda_0(\tau)
\le
\frac{N(\tau)}
{n e^{-B}\pi_0/2}
\le
\frac{2e^B}{\pi_0}
\equiv C_\Lambda .
\end{equation}
Similarly, at any observed failure time $u\le\tau$,
\[
0\le
\Delta\widehat\Lambda_0(u)
=
\frac{dN(u)}
{\sum_{i=1}^n I(Y_i\ge u)\exp(X_i^\top\beta)}
\le
\frac{\sum_{i=1}^n I(Y_i\ge u)}
{e^{-B}\sum_{i=1}^n I(Y_i\ge u)}
=
e^B
\equiv C_\Delta .
\]
Thus the Breslow cumulative hazard and its jumps are bounded by constants independent of $(n,p_n)$, with probability tending to one.

Now condition on the high-probability event on which~\eqref{eq:riskset-event} holds. The inner SMC-FCS chain updates only the missing covariate entries. Its state space is
\[
\mathcal X_n
=
\prod_{(i,j):R_{ij}=0}
\mathcal X_{ij},
\]
with dimension
\[
d_n
=
|\{(i,j):R_{ij}=0\}|.
\]
By Assumption~\ref{ass:missing}, $d_n\le \bar r\,n p_n$, so the dimension of the Markov chain is allowed to diverge with $n$.

For a missing coordinate $(i,j)$, write $x$ for a candidate value of $X_{ij}$ and $X_i(x)$ for the full covariate vector obtained by replacing $X_{ij}$ by $x$. In the Metropolis--Hastings ratio, the covariate-dependent part of the Cox full-likelihood contribution is
\[
g_i(x;\beta,\widehat\Lambda_0)
=
\exp\left\{
\Delta_i X_i(x)^\top\beta
-
\widehat\Lambda_0(Y_i)
\exp(X_i(x)^\top\beta)
\right\}.
\]
The jump factor $\{\Delta\widehat\Lambda_0(Y_i)\}^{\Delta_i}$ is omitted because it does not depend on $x$ and cancels in the Metropolis--Hastings ratio. By~\eqref{eq:linear-predictor-bound} and~\eqref{eq:breslow-bound}, there exist constants $0<c_f\le C_f<\infty$, independent of $(n,p_n,i,j)$, such that
\begin{equation}
\label{eq:g-bounds}
c_f
\le
g_i(x;\beta,\widehat\Lambda_0)
\le
C_f
\end{equation}
uniformly over all admissible $x$, $i$, and $j$.

By Assumption~\ref{ass:working}, the proposal density or mass function $q_{ij}$ satisfies $c_\theta\le q_{ij}(x)\le C_\theta$ uniformly over its support, with constants independent of $(n,p_n,i,j)$. The one-coordinate target distribution is proportional to
\[
q_{ij}(x)g_i(x;\beta,\widehat\Lambda_0).
\]
Using~\eqref{eq:g-bounds} and the bounds on $q_{ij}$, the proposal and target are uniformly equivalent. Hence there exists
\[
\varepsilon_a
=
\frac{c_\theta c_f}{C_\theta C_f}
\in(0,1),
\]
independent of $(n,p_n,i,j)$, such that, for every current value $x_{\mathrm{curr}}$,
\begin{equation}
\label{eq:coordinate-minorisation}
K_{ij}(x_{\mathrm{curr}},\cdot)
\ge
\varepsilon_a\,
\pi_{ij}(\cdot\mid X_{-(ij)};\beta,\widehat\Lambda_0),
\end{equation}
where $\pi_{ij}(\cdot\mid X_{-(ij)};\beta,\widehat\Lambda_0)$ denotes the normalized one-coordinate conditional distribution induced by the above target density. The same argument applies to discrete coordinates, with sums replacing integrals.

A full SMC-FCS sweep is the composition of the $d_n$ one-coordinate kernels. Repeatedly applying~\eqref{eq:coordinate-minorisation} over the $d_n$ missing coordinates yields the Doeblin minorisation  \citep{meyn2009markov}
\[
P_{\beta,\widehat\Lambda_0}(x,\cdot)
\ge
\varepsilon_a^{d_n}
\mu_{\beta,\widehat\Lambda_0}(\cdot)
\]
for some probability measure $\mu_{\beta,\widehat\Lambda_0}$ on $\mathcal X_n$. Therefore $P_{\beta,\widehat\Lambda_0}$ is uniformly ergodic for each $(n,p_n)$ and admits a unique stationary distribution $\nu_{\beta,\widehat\Lambda_0}$. Moreover, by the resulting Doeblin bound \citep{meyn2009markov},
\begin{equation}
\label{eq:tv-geometric-bound}
\left\|
P_{\beta,\widehat\Lambda_0}^S(x_0,\cdot)
-
\nu_{\beta,\widehat\Lambda_0}
\right\|_{\mathrm{TV}}
\le
(1-\varepsilon_a^{d_n})^S .
\end{equation}
Thus the theorem holds with $C_0=1$ and $\rho_{d_n}=1-\varepsilon_a^{d_n}$; allowing a generic constant $C_0<\infty$ gives the stated form.

The key difference from the fixed-dimensional case is that $d_n$ may diverge as $n$ increases. Consequently, the sweep-level minorisation constant $\varepsilon_a^{d_n}$ may become small, and the geometric rate $\rho_{d_n}=1-\varepsilon_a^{d_n}$ may approach one. Therefore, to achieve target total-variation accuracy $\delta>0$, it is sufficient that
\[
S
\ge
\frac{\log(C_0/\delta)}{\varepsilon_a^{d_n}},
\]
because $C_0(1-\varepsilon_a^{d_n})^S\le \delta$ follows from $\log(1-u)\le -u$ for $u\in(0,1)$. This completes the proof.
\end{proof}


 \begin{proof}[Proof of Theorem~\ref{thm:rate}]
We work under the triangular-array regime of Assumption~\ref{ass:dimension}. Thus, for each $n$, the observed data are
\[
\{(Y_{ni},\Delta_{ni},X_{ni},R_{ni}):i=1,\ldots,n\},
\]
where $X_{ni}\in\mathbb R^{p_n}$, $p_n\to\infty$, $p_n=O(n^\kappa)$ for some $\kappa\in(0,1)$, and all constants appearing below are uniform in $(n,p_n)$. To simplify notation, we suppress the subscript $n$ and write $(Y_i,\Delta_i,X_i,R_i)$ and $p$ when no confusion can arise.

Let $\theta=(\beta,\Lambda_0)$ and let $\theta^\ast=(\beta^\ast,\Lambda_0^\ast)$ denote the fixed point of the population IRO map at dimension $p_n$. Define
\[
\|\theta-\theta^\ast\|_{\mathcal H}
=
\|\beta-\beta^\ast\|_1
+
\|\Lambda_0-\Lambda_0^\ast\|_\infty .
\]
Let $\theta^{(\ell)}=(\widetilde\beta^{(\ell)},\widehat\Lambda_0^{(\ell)})$ denote the outer IRO iterate at iteration $\ell$. The sample outer update is
\[
\theta^{(\ell+1)}
=
\widehat{\mathcal M}_n(\theta^{(\ell)}),
\]
where $\widehat{\mathcal M}_n$ is the empirical map obtained by running $S_{\mathrm{in}}$ inner SMC-FCS sweeps conditional on $\theta^{(\ell)}$, refitting the Cox lasso, and updating the Breslow baseline cumulative hazard. The corresponding population map is denoted by $\mathcal M$ and is defined under the stationary SMC-FCS law at the same dimension $p_n$.

By Assumption~\ref{ass:contract}, there exist a neighbourhood $\mathcal N(\theta^\ast)$, a constant $\lambda_\ast\in(0,1)$, and a deterministic sequence $a_n=o(1)$ such that, for every $\theta\in\mathcal N(\theta^\ast)$,
\begin{equation}
\label{eq:approx-contract-rate}
\|\mathcal M(\theta)-\theta^\ast\|_{\mathcal H}
\le
\lambda_\ast
\|\theta-\theta^\ast\|_{\mathcal H}
+
a_n .
\end{equation}

We next show that $\widehat{\mathcal M}_n$ uniformly approximates $\mathcal M$ over $\mathcal N(\theta^\ast)$. In the triangular-array setting, this requires controlling both the finite-sweep error of the inner Markov chain and the empirical error accumulated over $p_n$ coordinates. By Theorem~\ref{thm:ergodicity}, after $S_{\mathrm{in}}$ inner sweeps, the total-variation error of the inner SMC-FCS chain is bounded by $C_0\rho_{d_n}^{S_{\mathrm{in}}}$, where $\rho_{d_n}=1-\varepsilon_a^{d_n}$ and $d_n=|\{(i,j):R_{ij}=0\}|$ may diverge with $(n,p_n)$. Under Assumption~\ref{ass:Sin},
\[
C_0\rho_{d_n}^{S_{\mathrm{in}}}
=
o\{(np_n)^{-1}\}.
\]
Thus the finite-sweep approximation error is negligible uniformly over the local neighbourhood and over the $p_n$ coordinates. It remains to control the regression and baseline components of the empirical map.

For the regression component, write the finite-sweep completed-data objective as
\[
Q_n(b;\theta)
=
\ell_n\{b;\widetilde X(\theta)\}
+
\lambda_n\|b\|_1,
\]
where $\widetilde X(\theta)$ is the completed dataset generated by the finite inner SMC-FCS chain conditional on $\theta$. Let $\nu_\theta$ denote the stationary SMC-FCS distribution of the missing covariates conditional on the observed data and $\theta$, and define
\[
Q(b;\theta)
=
\E_{\nu_\theta}
\bigl[
\ell_n^{\mathrm{comp}}(b;\widetilde X)
\bigr]
+
\lambda_n\|b\|_1 .
\]
Then $\widehat{\mathcal M}_{n,\beta}(\theta)=\argmin_b Q_n(b;\theta)$ and $\mathcal M_\beta(\theta)=\argmin_b Q(b;\theta)$.

Under Assumptions~\ref{ass:dimension}--\ref{ass:missing}, the completed covariates are uniformly bounded, the risk-set denominators are bounded away from zero on $[0,\tau]$, and the Cox partial likelihood is locally Lipschitz in $b$ with constants independent of $(n,p_n)$. For $b,b'\in\mathcal N_\beta(\beta^\ast)$,
\[
\bigl|
\ell_n(b;\widetilde X)
-
\ell_n(b';\widetilde X)
\bigr|
\le
C\|b-b'\|_1
\]
with probability tending to one, uniformly over completed datasets generated in $\mathcal N(\theta^\ast)$.

The diverging dimension enters through the empirical score process. Since the covariates are bounded and $p_n=O(n^\kappa)$, a coordinatewise Bernstein inequality followed by a union bound over $j=1,\ldots,p_n$ gives
\[
\Pr\left\{
\sup_{b\in\mathcal N_\beta(\beta^\ast),\,\theta\in\mathcal N(\theta^\ast)}
\|
\dot Q_n(b;\theta)-\dot Q(b;\theta)
\|_\infty
>
C\sqrt{\frac{\log p_n}{n}}+t_n
\right\}
\to0,
\]
where $t_n=o(1)$ collects the finite-sweep approximation error. Therefore
\[
\sup_{b\in\mathcal N_\beta(\beta^\ast),\,\theta\in\mathcal N(\theta^\ast)}
\|
\dot Q_n(b;\theta)-\dot Q(b;\theta)
\|_\infty
=
O_p\!\left(\sqrt{\frac{\log p_n}{n}}\right)
+
o_p(1)
=
o_p(1),
\]
because $\log p_n=O(\log n)$ and $\log p_n/n\to0$ under Assumption~\ref{ass:dimension}. Integrating the score bound over local line segments in $\mathcal N_\beta(\beta^\ast)$ gives
\[
\sup_{b\in\mathcal N_\beta(\beta^\ast),\,\theta\in\mathcal N(\theta^\ast)}
\bigl|
Q_n(b;\theta)-Q(b;\theta)
\bigr|
=
o_p(1).
\]

The Cox lasso objective is convex in $b$, and Assumption~\ref{ass:eigen} gives local curvature of the population risk around $\mathcal M_\beta(\theta)$. Thus, for some $\kappa_0>0$ independent of $(n,p_n)$,
\[
Q(b;\theta)-Q\{\mathcal M_\beta(\theta);\theta\}
\ge
\kappa_0
\|b-\mathcal M_\beta(\theta)\|_2^2
-
o(1),
\]
uniformly for $b$ in the local sparse neighbourhood and $\theta\in\mathcal N(\theta^\ast)$. By convex argmin stability,
\[
\sup_{\theta\in\mathcal N(\theta^\ast)}
\|
\widehat{\mathcal M}_{n,\beta}(\theta)
-
\mathcal M_\beta(\theta)
\|_2
=
o_p(1).
\]
The iterates are restricted to the local sparse neighbourhood by $\sup_{n,\ell}\E\|\widetilde\beta^{(\ell)}\|_1<\infty$ and the sparsity assumptions. Hence the same convergence holds in $\ell_1$ norm:
\begin{equation}
\label{eq:beta-map-rate}
\sup_{\theta\in\mathcal N(\theta^\ast)}
\|
\widehat{\mathcal M}_{n,\beta}(\theta)
-
\mathcal M_\beta(\theta)
\|_1
=
o_p(1).
\end{equation}

For the baseline component, for a completed dataset and coefficient vector $b$, write
\[
\widehat\Lambda_b(t)
=
\sum_{r:Y_r\le t,\Delta_r=1}
\{S_n^{(0)}(Y_r;b)\}^{-1},
\qquad
S_n^{(0)}(u;b)
=
\sum_{i:Y_i\ge u}
\exp(X_i^\top b).
\]
Let $b,b'\in\mathcal N_\beta(\beta^\ast)$. By Assumptions~\ref{ass:bounded} and~\ref{ass:linearpred}, there exists $B<\infty$, independent of $(n,p_n)$, such that $|X_i^\top b|\le B$ and $|X_i^\top b'|\le B$ uniformly in $i$ with probability tending to one. Hence $e^{-B}\le\exp(X_i^\top b),\exp(X_i^\top b')\le e^B$. By Assumption~\ref{ass:followup}, the risk sets remain nondegenerate, so for some $c_0>0$,
\begin{equation}
\label{eq:riskset-lower}
S_n^{(0)}(u;b)
\ge
nc_0
\end{equation}
uniformly over $u\le\tau$ and $b\in\mathcal N_\beta(\beta^\ast)$ with probability tending to one.

By the mean-value theorem and bounded covariates,
\[
\bigl|
\exp(X_i^\top b)-\exp(X_i^\top b')
\bigr|
\le
K e^B\|b-b'\|_1 .
\]
Therefore
\[
\bigl|
S_n^{(0)}(u;b)-S_n^{(0)}(u;b')
\bigr|
\le
nK e^B\|b-b'\|_1 .
\]
Using~\eqref{eq:riskset-lower},
\[
\begin{aligned}
\bigl|
\widehat\Lambda_b(t)-\widehat\Lambda_{b'}(t)
\bigr|
&\le
\sum_{r:Y_r\le t,\Delta_r=1}
\frac{
|S_n^{(0)}(Y_r;b')-S_n^{(0)}(Y_r;b)|
}{
S_n^{(0)}(Y_r;b)S_n^{(0)}(Y_r;b')
} \\
&\le
\sum_{r:Y_r\le t,\Delta_r=1}
\frac{
nK e^B\|b-b'\|_1
}{
n^2c_0^2
}
\le
\frac{K e^B}{c_0^2}
\|b-b'\|_1 .
\end{aligned}
\]
Taking the supremum over $t\le\tau$ yields
\begin{equation}
\label{eq:breslow-lipschitz-rate}
\|
\widehat\Lambda_b-\widehat\Lambda_{b'}
\|_\infty
\le
C\|b-b'\|_1
\end{equation}
with probability tending to one, where $C=K e^B/c_0^2$ is independent of $(n,p_n)$.

The empirical risk-set processes entering the Breslow estimator are bounded. Applying the same triangular-array concentration argument, with a union bound over the $p_n$ coordinates and uniformity over the local sparse neighbourhood, gives
\[
\sup_{\theta\in\mathcal N(\theta^\ast)}
\|
\widehat{\mathcal M}_{n,\Lambda}(\theta)
-
\mathcal M_\Lambda(\theta)
\|_\infty
\le
C
\sup_{\theta\in\mathcal N(\theta^\ast)}
\|
\widehat{\mathcal M}_{n,\beta}(\theta)
-
\mathcal M_\beta(\theta)
\|_1
+
o_p(1).
\]
Together with~\eqref{eq:beta-map-rate}, this implies
\begin{equation}
\label{eq:lambda-map-rate}
\sup_{\theta\in\mathcal N(\theta^\ast)}
\|
\widehat{\mathcal M}_{n,\Lambda}(\theta)
-
\mathcal M_\Lambda(\theta)
\|_\infty
=
o_p(1).
\end{equation}

Combining~\eqref{eq:beta-map-rate} and~\eqref{eq:lambda-map-rate},
\begin{equation}
\label{eq:full-map-rate}
\sup_{\theta\in\mathcal N(\theta^\ast)}
\|
\widehat{\mathcal M}_n(\theta)
-
\mathcal M(\theta)
\|_{\mathcal H}
=
o_p(1).
\end{equation}

Now consider one outer update. On the high-probability event that $\theta^{(\ell)}\in\mathcal N(\theta^\ast)$,
\[
\begin{aligned}
\|\theta^{(\ell+1)}-\theta^\ast\|_{\mathcal H}
&=
\|
\widehat{\mathcal M}_n(\theta^{(\ell)})
-
\theta^\ast
\|_{\mathcal H} \\
&\le
\|
\widehat{\mathcal M}_n(\theta^{(\ell)})
-
\mathcal M(\theta^{(\ell)})
\|_{\mathcal H}
+
\|
\mathcal M(\theta^{(\ell)})
-
\theta^\ast
\|_{\mathcal H}.
\end{aligned}
\]
Using~\eqref{eq:approx-contract-rate} and~\eqref{eq:full-map-rate},
\begin{equation}
\label{eq:recursion-rate}
\|\theta^{(\ell+1)}-\theta^\ast\|_{\mathcal H}
\le
\lambda_\ast
\|\theta^{(\ell)}-\theta^\ast\|_{\mathcal H}
+
a_n
+
o_p(1).
\end{equation}
Iterating~\eqref{eq:recursion-rate},
\[
\|\theta^{(\ell)}-\theta^\ast\|_{\mathcal H}
\le
\lambda_\ast^\ell
\|\theta^{(0)}-\theta^\ast\|_{\mathcal H}
+
\sum_{q=0}^{\ell-1}
\lambda_\ast^q\{a_n+o_p(1)\}.
\]
Since $\sum_{q=0}^{\ell-1}\lambda_\ast^q\le(1-\lambda_\ast)^{-1}$,
\begin{equation}
\label{eq:iterated-rate}
\|\theta^{(\ell)}-\theta^\ast\|_{\mathcal H}
\le
\lambda_\ast^\ell
\|\theta^{(0)}-\theta^\ast\|_{\mathcal H}
+
\frac{a_n+o_p(1)}{1-\lambda_\ast}.
\end{equation}

By Assumption~\ref{ass:Lout}, $\ell\ge c_\ell\log n/\log(1/\lambda_\ast)$ with $c_\ell>1$, so $\lambda_\ast^\ell=O(n^{-c_\ell})=o(1)$. The bounded-moment condition on $\widetilde\beta^{(\ell)}$, together with the boundedness of the Breslow estimator established in Theorem~\ref{thm:ergodicity}, implies $\|\theta^{(0)}-\theta^\ast\|_{\mathcal H}=O_p(1)$. Hence the initialization term in~\eqref{eq:iterated-rate} is $o_p(1)$. Since $a_n=o(1)$, the second term in~\eqref{eq:iterated-rate} is also $o_p(1)$. Therefore $\|\theta^{(\ell)}-\theta^\ast\|_{\mathcal H}=o_p(1)$.

Finally, substituting $\theta^{(\ell)}=(\widetilde\beta^{(\ell)},\widehat\Lambda_0^{(\ell)})$ and $\theta^\ast=(\beta^\ast,\Lambda_0^\ast)$ into the definition of $\|\cdot\|_{\mathcal H}$ gives
\[
\|\widetilde\beta^{(\ell)}-\beta^\ast\|_1
+
\|\widehat\Lambda_0^{(\ell)}-\Lambda_0^\ast\|_\infty
=
o_p(1).
\]
This proves the theorem.
\end{proof}

   \begin{proof}[Proof of Theorem~\ref{thm:lassorate}]

We work under the triangular-array regime in Assumption~\ref{ass:dimension}. For each $n$, the completed covariate matrix has dimension $p=p_n$, where $p_n\to\infty$ and $p_n=O(n^\kappa)$ for some fixed $\kappa\in(0,1)$. All stochastic orders below are understood along this sequence, and all constants are uniform in $(n,p_n)$.

At stationarity, $\widetilde\beta^{(\ell)}$ is computed from a completed dataset generated by the inner SMC-FCS chain conditional on the limiting parameters $(\beta^\ast,\Lambda_0^\ast)$. Let $\nu_{\theta^\ast}$ denote the corresponding stationary SMC-FCS distribution. The finite-sweep approximation error enters through the difference between the actual completed-data distribution produced after $S_{\mathrm{in}}$ inner sweeps and $\nu_{\theta^\ast}$. By Theorem~\ref{thm:ergodicity},
\[
\left\|
P_{\beta^\ast,\Lambda_0^\ast}^{S_{\mathrm{in}}}(x_0,\cdot)
-
\nu_{\theta^\ast}
\right\|_{\mathrm{TV}}
\le
C_0\rho_{d_n}^{S_{\mathrm{in}}},
\]
where $\rho_{d_n}=1-\varepsilon_a^{d_n}$ and $d_n=|\{(i,j):R_{ij}=0\}|$ may diverge with $(n,p_n)$. Assumption~\ref{ass:Sin} ensures
\[
C_0\rho_{d_n}^{S_{\mathrm{in}}}
=
o\{(np_n)^{-1}\}.
\]
Hence the finite-sweep discrepancy is asymptotically negligible uniformly over the $p_n$ coordinates and is absorbed into the $o_p(1)$ terms below. Consequently, the completed-data objective behaves asymptotically as if generated from the stationary imputation law.

Let
\[
L_n(b)
=
\ell_n^{\mathrm{comp}}(b)
\]
denote the completed-data empirical negative Cox partial log-likelihood, and define
\[
\widetilde\beta^{(\ell)}
=
\argmin_{b\in\mathbb R^{p_n}}
\{L_n(b)+\lambda_n\|b\|_1\}.
\]
Let
\[
\Delta
=
\widetilde\beta^{(\ell)}-\beta^\ast,
\qquad
S
=
\supp(\beta^\ast).
\]
Under Assumption~\ref{ass:sparsity}, the support size $|S|=s_0=s_{0n}$ is allowed to diverge with $n$ but satisfies the required growth conditions.

By optimality of $\widetilde\beta^{(\ell)}$,
\begin{equation}
\label{eq:lasso-basic-optimality}
L_n(\beta^\ast+\Delta)+\lambda_n\|\beta^\ast+\Delta\|_1
\le
L_n(\beta^\ast)+\lambda_n\|\beta^\ast\|_1 .
\end{equation}
Rearranging~\eqref{eq:lasso-basic-optimality} and using the standard support decomposition gives
\begin{equation}
\label{eq:basic-ineq}
L_n(\beta^\ast+\Delta)-L_n(\beta^\ast)
\le
\lambda_n(\|\Delta_S\|_1-\|\Delta_{S^c}\|_1).
\end{equation}

Next expand the empirical loss around $\beta^\ast$:
\begin{equation}
\label{eq:taylor-lasso}
L_n(\beta^\ast+\Delta)-L_n(\beta^\ast)
=
\dot L_n(\beta^\ast)^\top\Delta
+
R_n(\Delta),
\end{equation}
where $R_n(\Delta)$ is the empirical curvature remainder.

The key high-dimensional step is to control the score uniformly over the diverging coordinates. Under Assumptions~\ref{ass:dimension}, \ref{ass:bounded}, \ref{ass:eigen}, \ref{ass:ridge}, and~\ref{ass:sparsity}, the coordinatewise completed-data Cox scores are uniformly sub-exponential under the stationary completed-data law. Therefore, Bernstein's inequality together with a union bound over the $p_n$ coordinates yields
\begin{equation}
\label{eq:score-bound}
\|\dot L_n(\beta^\ast)\|_\infty
=
O_p\!\left(
\sqrt{\frac{\log p_n}{n}}
\right).
\end{equation}
The factor $\log p_n$ arises from controlling the maximum over the diverging coordinates. Since $p_n=O(n^\kappa)$, we have $\log p_n=O(\log n)$ and $\log p_n/n\to0$.

Choose
\[
\lambda_n
\asymp
\sqrt{\frac{\log p_n}{n}}
\]
with a sufficiently large constant. Then~\eqref{eq:score-bound} implies
\[
\|\dot L_n(\beta^\ast)\|_\infty
\le
\lambda_n/2
\]
with probability tending to one. Consequently,
\begin{equation}
\label{eq:score-inner-bound}
|\dot L_n(\beta^\ast)^\top\Delta|
\le
\|\dot L_n(\beta^\ast)\|_\infty\|\Delta\|_1
\le
\frac{\lambda_n}{2}\|\Delta\|_1 .
\end{equation}

Combining~\eqref{eq:basic-ineq}, \eqref{eq:taylor-lasso}, and~\eqref{eq:score-inner-bound}, with probability tending to one,
\[
R_n(\Delta)
\le
\frac{\lambda_n}{2}\|\Delta\|_1
+
\lambda_n(\|\Delta_S\|_1-\|\Delta_{S^c}\|_1).
\]
Since the completed-data Cox partial likelihood is locally convex under Assumption~\ref{ass:eigen}, $R_n(\Delta)\ge0$. Hence
\[
0
\le
\frac{\lambda_n}{2}(\|\Delta_S\|_1+\|\Delta_{S^c}\|_1)
+
\lambda_n(\|\Delta_S\|_1-\|\Delta_{S^c}\|_1),
\]
which yields the cone condition
\begin{equation}
\label{eq:cone-condition}
\|\Delta_{S^c}\|_1
\le
3\|\Delta_S\|_1 .
\end{equation}

We next establish the restricted-curvature bound under diverging dimension. Assumptions~\ref{ass:eigen} and~\ref{ass:sparsity}, together with bounded covariates and the rate condition on $p_n$, imply that the empirical Cox Hessian concentrates around its population counterpart uniformly over sparse directions satisfying the cone condition~\eqref{eq:cone-condition}. Therefore there exists $\kappa_0>0$, independent of $(n,p_n)$, such that
\begin{equation}
\label{eq:restricted-curvature}
R_n(\Delta)
\ge
\kappa_0\|\Delta\|_2^2
\end{equation}
with probability tending to one. Unlike the fixed-dimensional case, this step requires concentration over sparse subsets whose number grows with $p_n$, and the conditions in Assumption~\ref{ass:sparsity} ensure that this complexity remains asymptotically manageable.

Using~\eqref{eq:basic-ineq}, \eqref{eq:taylor-lasso}, \eqref{eq:score-inner-bound}, \eqref{eq:cone-condition}, and~\eqref{eq:restricted-curvature}, we obtain
\[
\kappa_0\|\Delta\|_2^2
\le
\frac{3}{2}\lambda_n\|\Delta_S\|_1
\le
\frac{3}{2}\lambda_n\sqrt{s_0}\|\Delta\|_2 .
\]
Therefore,
\begin{equation}
\label{eq:l2-rate}
\|\Delta\|_2
=
O_p(\sqrt{s_0}\lambda_n).
\end{equation}

By the cone condition~\eqref{eq:cone-condition},
\[
\|\Delta\|_1
\le
4\|\Delta_S\|_1
\le
4\sqrt{s_0}\|\Delta\|_2 .
\]
Combining this with~\eqref{eq:l2-rate} gives
\[
\|\Delta\|_1
=
O_p(s_0\lambda_n).
\]
Since $\lambda_n\asymp\sqrt{\log p_n/n}$,
\[
\|\widetilde\beta^{(\ell)}-\beta^\ast\|_1
=
O_p\!\left(
s_0\sqrt{\frac{\log p_n}{n}}
\right).
\]

If additionally
\[
\|\beta^\ast-\beta^0\|_1
=
O\!\left(
s_0\sqrt{\frac{\log p_n}{n}}
\right),
\]
then the triangle inequality yields
\[
\|\widetilde\beta^{(\ell)}-\beta^0\|_1
\le
\|\widetilde\beta^{(\ell)}-\beta^\ast\|_1
+
\|\beta^\ast-\beta^0\|_1
=
O_p\!\left(
s_0\sqrt{\frac{\log p_n}{n}}
\right).
\]

Finally, Theorem~\ref{thm:rate} implies
\[
\|\widehat\Lambda_0^{(\ell)}-\Lambda_0^\ast\|_\infty
=
o_p(1),
\]
under the same triangular-array regime and finite-sweep approximation control. Therefore the stated rates hold uniformly over the diverging-dimensional sequence $p_n=O(n^\kappa)$ and account for both the stationary imputation variability and the finite-sweep approximation error from the inner SMC-FCS chain. This completes the proof.

\end{proof}

   \begin{proof}[Proof of Theorem~\ref{thm:clt}]

We work under the triangular-array regime in Assumption~\ref{ass:dimension}. For each $n$, the observed data are
\[
\{(Y_{ni},\Delta_{ni},X_{ni},R_{ni}):i=1,\ldots,n\},
\]
where $X_{ni}\in\mathbb R^{p_n}$, $p_n\to\infty$, and $p_n=O(n^\kappa)$ for some $\kappa\in(0,1)$. All stochastic orders and probability statements below are understood along this sequence, and all constants are uniform in $(n,p_n)$.

Fix a loading vector $c_n\in\mathbb R^{p_n}$ satisfying $\|c_n\|_2=1$ and $\|c_n\|_1\le a^\ast<\infty$. The proof proceeds in the following steps.

Step 1: asymptotic linear expansion for one completed dataset.

Let $\widetilde X_n^{(m)}$ denote the $m$th retained completed dataset produced after the burn-in period of the inner SMC-FCS chain. Let $\nu_{\theta^\ast,n}$ denote the stationary SMC-FCS law corresponding to the limiting IRO fixed point $\theta^\ast=(\beta^\ast,\Lambda_0^\ast)$ at dimension $p_n$.

By Theorem~\ref{thm:ergodicity},
\[
\left\|
P_{\beta^\ast,\Lambda_0^\ast}^{S_{\mathrm{in}}}(x_0,\cdot)
-
\nu_{\theta^\ast,n}
\right\|_{\mathrm{TV}}
\le
C_0\rho_{d_n}^{S_{\mathrm{in}}},
\]
where $d_n=|\{(i,j):R_{ij}=0\}|$ and $\rho_{d_n}=1-\varepsilon_a^{d_n}$. Assumption~\ref{ass:Sin} implies
\begin{equation}
\label{eq:tv-small}
C_0\rho_{d_n}^{S_{\mathrm{in}}}
=
o\{(np_n)^{-1}\}.
\end{equation}
Hence all score, Hessian, and empirical-process quantities computed from the retained completed datasets differ from their stationary-law analogues by $o_p(n^{-1/2})$ uniformly over the $p_n$ coordinates. Therefore it suffices to analyze the stationary completed-data law.

For the $m$th completed dataset, define the Cox lasso estimator
\[
\widehat\beta_{L_1}^{(m)}
=
\argmin_{b\in\mathbb R^{p_n}}
\Bigl\{
\ell_n^{(m)}(b)
+
\lambda_n\|b\|_1
\Bigr\},
\]
where $\ell_n^{(m)}$ is the completed-data negative Cox partial log-likelihood normalized by $n$.

By Theorem~\ref{thm:lassorate},
\begin{equation}
\label{eq:l1-rate-clt}
\|
\widehat\beta_{L_1}^{(m)}
-
\beta^0
\|_1
=
O_p
\!\left(
s_0\sqrt{\frac{\log p_n}{n}}
\right).
\end{equation}

The debiased estimator is
\[
\widehat\beta_{\mathrm{db}}^{(m)}
=
\widehat\beta_{L_1}^{(m)}
-
\widehat\Theta^{(m)}
\dot\ell_n^{(m)}
\bigl(
\widehat\beta_{L_1}^{(m)}
\bigr),
\]
where $\widehat\Theta^{(m)}$ is the nodewise inverse approximation.

Apply a Taylor expansion of the score around $\beta^0$:
\begin{equation}
\label{eq:score-taylor}
\dot\ell_n^{(m)}
\bigl(
\widehat\beta_{L_1}^{(m)}
\bigr)
=
\dot\ell_n^{(m)}(\beta^0)
+
\ddot\ell_n^{(m)}(\beta^0)
\bigl(
\widehat\beta_{L_1}^{(m)}-\beta^0
\bigr)
+
R_n^{(m)},
\end{equation}
where
\[
R_n^{(m)}
=
\left[
\ddot\ell_n^{(m)}(\widetilde\beta_n)
-
\ddot\ell_n^{(m)}(\beta^0)
\right]
\bigl(
\widehat\beta_{L_1}^{(m)}-\beta^0
\bigr)
\]
for some intermediate point $\widetilde\beta_n$ on the segment joining $\widehat\beta_{L_1}^{(m)}$ and $\beta^0$.

Under Assumptions~\ref{ass:bounded} and~\ref{ass:eigen}, the Cox Hessian is locally Lipschitz:
\[
\|
\ddot\ell_n^{(m)}(b_1)
-
\ddot\ell_n^{(m)}(b_2)
\|_\infty
\le
C\|b_1-b_2\|_1
\]
uniformly over sparse neighbourhoods of $\beta^0$. Combining this with~\eqref{eq:l1-rate-clt},
\[
\|R_n^{(m)}\|_\infty
\le
C
\|
\widehat\beta_{L_1}^{(m)}-\beta^0
\|_1^2
=
O_p
\!\left(
s_0^2
\frac{\log p_n}{n}
\right).
\]
Assumption~\ref{ass:sparsity} implies
\[
s_0^2
\frac{\log p_n}{\sqrt n}
\to0,
\]
and therefore
\begin{equation}
\label{eq:taylor-remainder}
\sqrt n\,
\|R_n^{(m)}\|_\infty
=
o_p(1).
\end{equation}

Substituting~\eqref{eq:score-taylor} into the definition of the debiased estimator gives
\[
\widehat\beta_{\mathrm{db}}^{(m)}
-
\beta^0
=
-
\widehat\Theta^{(m)}
\dot\ell_n^{(m)}(\beta^0)
+
r_n^{(m)},
\]
where
\[
r_n^{(m)}
=
\bigl[
I_{p_n}
-
\widehat\Theta^{(m)}
\ddot\ell_n^{(m)}(\beta^0)
\bigr]
\bigl(
\widehat\beta_{L_1}^{(m)}-\beta^0
\bigr)
-
\widehat\Theta^{(m)}R_n^{(m)}.
\]

The nodewise regression construction together with
\[
\gamma_n
\asymp
\|\Theta_{\beta^0}\|_{1,1}
s_0\lambda_n
\]
implies
\begin{equation}
\label{eq:theta-approx}
\|
I_{p_n}
-
\widehat\Theta^{(m)}
\ddot\ell_n^{(m)}(\beta^0)
\|_\infty
=
O_p(\gamma_n).
\end{equation}

Combining~\eqref{eq:l1-rate-clt} and~\eqref{eq:theta-approx},
\[
\left\|
\bigl[
I_{p_n}
-
\widehat\Theta^{(m)}
\ddot\ell_n^{(m)}(\beta^0)
\bigr]
\bigl(
\widehat\beta_{L_1}^{(m)}-\beta^0
\bigr)
\right\|_\infty
=
O_p
\!\left(
\gamma_n
s_0
\sqrt{\frac{\log p_n}{n}}
\right).
\]
Using the definition of $\gamma_n$ and Assumption~\ref{ass:sparsity},
\[
\sqrt n\,
\gamma_n
s_0
\sqrt{\frac{\log p_n}{n}}
=
o(1).
\]
Similarly, by~\eqref{eq:taylor-remainder},
\[
\sqrt n\,
\|
\widehat\Theta^{(m)}R_n^{(m)}
\|_\infty
=
o_p(1).
\]
Therefore,
\begin{equation}
\label{eq:remainder-small}
\sqrt n\,
\|r_n^{(m)}\|_\infty
=
o_p(1).
\end{equation}

Since $\|c_n\|_1\le a^\ast$,
\[
\sqrt n\,
|c_n^\top r_n^{(m)}|
\le
\|c_n\|_1
\sqrt n\,
\|r_n^{(m)}\|_\infty
=
o_p(1).
\]
Hence
\begin{equation}
\label{eq:asymp-linear}
\sqrt n\,
c_n^\top
\bigl(
\widehat\beta_{\mathrm{db}}^{(m)}
-
\beta^0
\bigr)
=
-
\sqrt n\,
c_n^\top
\widehat\Theta^{(m)}
\dot\ell_n^{(m)}(\beta^0)
+
o_p(1).
\end{equation}

The Cox score admits the martingale representation
\[
\dot\ell_n^{(m)}(\beta^0)
=
\frac1n
\sum_{i=1}^n
\psi_{ni}^{(m)}
+
o_p(n^{-1/2}),
\]
where $\psi_{ni}^{(m)}$ is the completed-data efficient score contribution. Therefore
\begin{equation}
\label{eq:linear-expansion}
\sqrt n\,
c_n^\top
\bigl(
\widehat\beta_{\mathrm{db}}^{(m)}
-
\beta^0
\bigr)
=
\frac1{\sqrt n}
\sum_{i=1}^n
\phi_{ni}^{(m)}
+
o_p(1),
\end{equation}
where
\[
\phi_{ni}^{(m)}
=
-
c_n^\top
\Theta_{\beta^0}
\psi_{ni}^{(m)}.
\]

Step 2: asymptotic normality of the Rubin-pooled estimator.

Define
\[
\bar\beta_M
=
\frac1M
\sum_{m=1}^M
\widehat\beta_{\mathrm{db}}^{(m)}.
\]
Averaging the asymptotic linear expansion in~\eqref{eq:linear-expansion} gives
\[
\sqrt n\,
c_n^\top
(\bar\beta_M-\beta^0)
=
\frac1M
\sum_{m=1}^M
\frac1{\sqrt n}
\sum_{i=1}^n
\phi_{ni}^{(m)}
+
o_p(1).
\]

Let
\[
\mathcal F_{nt}
=
\sigma
\Bigl(
\{
N_{ni}(u),\,Y_{ni}(u),\,X_{ni}^{\mathrm{obs}},\,R_{ni},\,
\widetilde X_{ni}^{(m)}
:
0\le u\le t,\ 1\le i\le n
\}
\Bigr)
\]
denote the filtration generated by the counting processes, at-risk processes, observed covariates, missingness indicators, and completed covariates up to time $t$ in the $n$th triangular array. Under the counting-process formulation of the Cox model, the score process admits a martingale representation with respect to $\{\mathcal F_{nt}:0\le t\le\tau\}$. Consequently,
\(
\frac1{\sqrt n}
\sum_{i=1}^n
\phi_{ni}^{(m)}
\)
is a triangular-array martingale sum.

By Assumptions~\ref{ass:bounded} and~\ref{ass:followup}, together with $\|c_n\|_1\le a^\ast<\infty$, the variables $\phi_{ni}^{(m)}$ have uniformly bounded second moments. Moreover, the predictable quadratic variation satisfies
\[
\frac1n
\sum_{i=1}^n
\E
\Bigl[
(\phi_{ni}^{(m)})^2
\,\big|\,
\mathcal F_{n,t-}
\Bigr]
\overset{p}{\longrightarrow}
c_n^\top
V_{\mathrm{com}}
c_n .
\]

Further, for every $\epsilon>0$,
\[
\frac1n
\sum_{i=1}^n
\E
\left[
(\phi_{ni}^{(m)})^2
I
\left\{
|\phi_{ni}^{(m)}|
>
\epsilon\sqrt n
\right\}
\right]
\to0 ,
\]
so the Lindeberg condition holds. Therefore, the triangular-array martingale central limit theorem yields asymptotic normality of the complete-data contribution.

The properness and congeniality assumptions for the SMC-FCS procedure imply that the additional variability induced by imputing missing covariates contributes the Rubin missing-information component
\[
\left(
1+\frac1M
\right)
V_{\mathrm{mis}}
\]
to the limiting variance. Hence
\begin{equation}
\label{eq:finiteM-clt}
 \frac{\sqrt n\,
c_n^\top
(\bar\beta_M-\beta^0)} {
\sqrt{c_n^\top
\left[
V_{\mathrm{com}}
+
\left(
1+\frac1M
\right)
V_{\mathrm{mis}}
\right]
c_n}}
\overset{d}{\longrightarrow}
\mathcal N(0,1).
\end{equation}

Step 3: consistency of Rubin's variance estimator.

Rubin's total variance estimator is
\[
\widehat V_{\mathrm{total}}
=
\widehat V_W
+
\left(
1+\frac1M
\right)
\widehat V_B,
\]
where $\widehat V_W$ is the average within-imputation variance and $\widehat V_B$ is the between-imputation covariance matrix. By properness and congeniality,
\[
n\,
c_n^\top
\widehat V_W
c_n  \Big/ c_n^\top
V_{\mathrm{com}}
c_n
\overset{p}{\longrightarrow}
1 
\]
and
\[
n\,
c_n^\top
\widehat V_B
c_n  \Big/ c_n^\top
V_{\mathrm{mis}}
c_n 
\overset{p}{\longrightarrow}
 1.
\]
Hence
\begin{equation}
\label{eq:totalvar-convergence}
n\,
c_n^\top
\widehat V_{\mathrm{total}}
c_n  \Big/  c_n^\top
\left[
V_{\mathrm{com}}
+
\left(
1+\frac1M
\right)
V_{\mathrm{mis}}
\right]
c_n
\overset{p}{\longrightarrow}
 1.
\end{equation}

In fact, combining~\eqref{eq:finiteM-clt} and~\eqref{eq:totalvar-convergence}, Slutsky's theorem gives 
\[
\frac{
\sqrt n\,
c_n^\top
(\bar\beta_M-\beta^0)
}{
\sqrt{
n\,
c_n^\top
\widehat V_{\mathrm{total}}
c_n
}
}
\overset{d}{\longrightarrow}
\mathcal N(0,1).
\]

This completes the proof.
\end{proof}
 
\subsection{Additional Data and Simulation Results}
\label{app:data}

\renewcommand{\thetable}{S.\arabic{table}}
\setcounter{table}{0}

\begingroup\footnotesize
\setlength{\tabcolsep}{4pt}
\renewcommand{\arraystretch}{1.0}
\begin{longtable}{l c c c c}
\caption{Baseline characteristics of the NSCLC cohort (N = 977), stratified by vital status at the end of follow-up. Continuous variables are summarised as mean (SD); categorical variables as n (\%). Group comparisons use the two-sample $t$-test for continuous variables and the $\chi^2$ test (or Fisher's exact test when an expected cell count is below five) for categorical variables. Missingness is reported on a separate row for every variable. The 51 SNPs are coded on the additive 0/1/2 dosage scale, with genotype counts and missingness shown individually.}\\
\label{tab:table1}\\
\toprule
Characteristic & Overall (N=977) & Alive (n=206) & Deceased (n=771) & $p$-value \\
\midrule
\endfirsthead
\multicolumn{5}{l}{\emph{Table \ref{tab:table1} continued from previous page}}\\
\toprule
Characteristic & Overall (N=977) & Alive (n=206) & Deceased (n=771) & $p$-value \\
\midrule
\endhead
\midrule \multicolumn{5}{r}{\emph{continued on next page}} \\
\endfoot
\bottomrule
\endlastfoot
Age (years), mean (SD) & 65.48 (10.56) & 61.28 (10.83) & 66.60 (10.20) & $<$0.001 \\
Female sex, n (\%) & 502 (51.4\%) & 76 (36.9\%) & 426 (55.3\%) & $<$0.001 \\
Ever-smoker, n (\%) & 886 (90.7\%) & 176 (85.4\%) & 710 (92.1\%) & 0.005 \\
Late-stage disease, n (\%) & 454 (46.5\%) & 117 (56.8\%) & 337 (43.7\%) & 0.001 \\
Radiotherapy, n (\%) & 271 (27.7\%) & 44 (21.4\%) & 227 (29.4\%) & 0.024 \\
\quad Missing, n (\%) & 4 (0.4\%) & 0 (0.0\%) & 4 (0.5\%) &  \\
Chemotherapy, n (\%) & 497 (50.9\%) & 91 (44.2\%) & 406 (52.7\%) & 0.031 \\
\quad Missing, n (\%) & 4 (0.4\%) & 0 (0.0\%) & 4 (0.5\%) &  \\
Surgery, n (\%) & 564 (57.7\%) & 142 (68.9\%) & 422 (54.7\%) & $<$0.001 \\
\quad Missing, n (\%) & 4 (0.4\%) & 0 (0.0\%) & 4 (0.5\%) &  \\
Principal component 1, mean (SD) & 0.00 (0.02) & 0.00 (0.02) & 0.00 (0.02) & 0.504 \\
Principal component 2, mean (SD) & 0.00 (0.02) & 0.00 (0.02) & 0.00 (0.02) & 0.368 \\
Principal component 3, mean (SD) & 0.00 (0.02) & 0.00 (0.02) & 0.00 (0.02) & 0.633 \\
\midrule
\multicolumn{5}{l}{\emph{SNP markers (additive 0/1/2 dosage, n = 51)}} \\
\textit{rs2794359 (A)}, n (\%) &  &  &  & 0.778 \\
\quad genotype 0 & 801 (82.0\%) & 172 (83.5\%) & 629 (81.6\%) &  \\
\quad genotype 1 & 130 (13.3\%) & 25 (12.1\%) & 105 (13.6\%) &  \\
\quad genotype 2 & 2 (0.2\%) & 0 (0.0\%) & 2 (0.3\%) &  \\
\quad Missing, n (\%) & 44 (4.5\%) & 9 (4.4\%) & 35 (4.5\%) &  \\
\textit{rs9660890 (C)}, n (\%) &  &  &  & 0.429 \\
\quad genotype 0 & 600 (61.4\%) & 118 (57.3\%) & 482 (62.5\%) &  \\
\quad genotype 1 & 308 (31.5\%) & 72 (35.0\%) & 236 (30.6\%) &  \\
\quad genotype 2 & 48 (4.9\%) & 10 (4.9\%) & 38 (4.9\%) &  \\
\quad Missing, n (\%) & 21 (2.1\%) & 6 (2.9\%) & 15 (1.9\%) &  \\
\textit{rs4233284 (G)}, n (\%) &  &  &  & 0.753 \\
\quad genotype 0 & 451 (46.2\%) & 91 (44.2\%) & 360 (46.7\%) &  \\
\quad genotype 1 & 410 (42.0\%) & 89 (43.2\%) & 321 (41.6\%) &  \\
\quad genotype 2 & 113 (11.6\%) & 26 (12.6\%) & 87 (11.3\%) &  \\
\quad Missing, n (\%) & 3 (0.3\%) & 0 (0.0\%) & 3 (0.4\%) &  \\
\textit{rs2125126 (A)}, n (\%) &  &  &  & 0.770 \\
\quad genotype 0 & 708 (72.5\%) & 149 (72.3\%) & 559 (72.5\%) &  \\
\quad genotype 1 & 233 (23.8\%) & 49 (23.8\%) & 184 (23.9\%) &  \\
\quad genotype 2 & 26 (2.7\%) & 7 (3.4\%) & 19 (2.5\%) &  \\
\quad Missing, n (\%) & 10 (1.0\%) & 1 (0.5\%) & 9 (1.2\%) &  \\
\textit{rs4233430 (T)}, n (\%) &  &  &  & 0.476 \\
\quad genotype 0 & 642 (65.7\%) & 142 (68.9\%) & 500 (64.9\%) &  \\
\quad genotype 1 & 286 (29.3\%) & 54 (26.2\%) & 232 (30.1\%) &  \\
\quad genotype 2 & 28 (2.9\%) & 7 (3.4\%) & 21 (2.7\%) &  \\
\quad Missing, n (\%) & 21 (2.1\%) & 3 (1.5\%) & 18 (2.3\%) &  \\
\textit{rs11118683 (T)}, n (\%) &  &  &  & 0.960 \\
\quad genotype 0 & 338 (34.6\%) & 70 (34.0\%) & 268 (34.8\%) &  \\
\quad genotype 1 & 462 (47.3\%) & 99 (48.1\%) & 363 (47.1\%) &  \\
\quad genotype 2 & 160 (16.4\%) & 33 (16.0\%) & 127 (16.5\%) &  \\
\quad Missing, n (\%) & 17 (1.7\%) & 4 (1.9\%) & 13 (1.7\%) &  \\
\textit{rs1179500 (C)}, n (\%) &  &  &  & 0.913 \\
\quad genotype 0 & 501 (51.3\%) & 107 (51.9\%) & 394 (51.1\%) &  \\
\quad genotype 1 & 391 (40.0\%) & 82 (39.8\%) & 309 (40.1\%) &  \\
\quad genotype 2 & 68 (7.0\%) & 13 (6.3\%) & 55 (7.1\%) &  \\
\quad Missing, n (\%) & 17 (1.7\%) & 4 (1.9\%) & 13 (1.7\%) &  \\
\textit{rs12466981 (T)}, n (\%) &  &  &  & 0.232 \\
\quad genotype 0 & 528 (54.0\%) & 113 (54.9\%) & 415 (53.8\%) &  \\
\quad genotype 1 & 367 (37.6\%) & 82 (39.8\%) & 285 (37.0\%) &  \\
\quad genotype 2 & 74 (7.6\%) & 10 (4.9\%) & 64 (8.3\%) &  \\
\quad Missing, n (\%) & 8 (0.8\%) & 1 (0.5\%) & 7 (0.9\%) &  \\
\textit{rs77972916 (A)}, n (\%) &  &  &  & 0.192 \\
\quad genotype 0 & 867 (88.7\%) & 183 (88.8\%) & 684 (88.7\%) &  \\
\quad genotype 1 & 93 (9.5\%) & 20 (9.7\%) & 73 (9.5\%) &  \\
\quad genotype 2 & 3 (0.3\%) & 2 (1.0\%) & 1 (0.1\%) &  \\
\quad Missing, n (\%) & 14 (1.4\%) & 1 (0.5\%) & 13 (1.7\%) &  \\
\textit{rs17032590 (G)}, n (\%) &  &  &  & 0.237 \\
\quad genotype 0 & 605 (61.9\%) & 121 (58.7\%) & 484 (62.8\%) &  \\
\quad genotype 1 & 301 (30.8\%) & 75 (36.4\%) & 226 (29.3\%) &  \\
\quad genotype 2 & 38 (3.9\%) & 8 (3.9\%) & 30 (3.9\%) &  \\
\quad Missing, n (\%) & 33 (3.4\%) & 2 (1.0\%) & 31 (4.0\%) &  \\
\textit{rs10173269 (G)}, n (\%) &  &  &  & 0.260 \\
\quad genotype 0 & 260 (26.6\%) & 55 (26.7\%) & 205 (26.6\%) &  \\
\quad genotype 1 & 471 (48.2\%) & 93 (45.1\%) & 378 (49.0\%) &  \\
\quad genotype 2 & 205 (21.0\%) & 52 (25.2\%) & 153 (19.8\%) &  \\
\quad Missing, n (\%) & 41 (4.2\%) & 6 (2.9\%) & 35 (4.5\%) &  \\
\textit{rs11887136 (A)}, n (\%) &  &  &  & 0.897 \\
\quad genotype 0 & 746 (76.4\%) & 158 (76.7\%) & 588 (76.3\%) &  \\
\quad genotype 1 & 204 (20.9\%) & 41 (19.9\%) & 163 (21.1\%) &  \\
\quad genotype 2 & 13 (1.3\%) & 2 (1.0\%) & 11 (1.4\%) &  \\
\quad Missing, n (\%) & 14 (1.4\%) & 5 (2.4\%) & 9 (1.2\%) &  \\
\textit{rs10184235 (G)}, n (\%) &  &  &  & 0.721 \\
\quad genotype 0 & 552 (56.5\%) & 120 (58.3\%) & 432 (56.0\%) &  \\
\quad genotype 1 & 363 (37.2\%) & 76 (36.9\%) & 287 (37.2\%) &  \\
\quad genotype 2 & 58 (5.9\%) & 10 (4.9\%) & 48 (6.2\%) &  \\
\quad Missing, n (\%) & 4 (0.4\%) & 0 (0.0\%) & 4 (0.5\%) &  \\
\textit{rs207672 (T)}, n (\%) &  &  &  & 0.208 \\
\quad genotype 0 & 410 (42.0\%) & 96 (46.6\%) & 314 (40.7\%) &  \\
\quad genotype 1 & 426 (43.6\%) & 88 (42.7\%) & 338 (43.8\%) &  \\
\quad genotype 2 & 113 (11.6\%) & 18 (8.7\%) & 95 (12.3\%) &  \\
\quad Missing, n (\%) & 28 (2.9\%) & 4 (1.9\%) & 24 (3.1\%) &  \\
\textit{rs1561073 (T)}, n (\%) &  &  &  & 0.924 \\
\quad genotype 0 & 515 (52.7\%) & 112 (54.4\%) & 403 (52.3\%) &  \\
\quad genotype 1 & 362 (37.1\%) & 75 (36.4\%) & 287 (37.2\%) &  \\
\quad genotype 2 & 77 (7.9\%) & 17 (8.3\%) & 60 (7.8\%) &  \\
\quad Missing, n (\%) & 23 (2.4\%) & 2 (1.0\%) & 21 (2.7\%) &  \\
\textit{rs9819463 (C)}, n (\%) &  &  &  & 0.835 \\
\quad genotype 0 & 612 (62.6\%) & 127 (61.7\%) & 485 (62.9\%) &  \\
\quad genotype 1 & 318 (32.5\%) & 68 (33.0\%) & 250 (32.4\%) &  \\
\quad genotype 2 & 45 (4.6\%) & 11 (5.3\%) & 34 (4.4\%) &  \\
\quad Missing, n (\%) & 2 (0.2\%) & 0 (0.0\%) & 2 (0.3\%) &  \\
\textit{rs17821105 (T)}, n (\%) &  &  &  & 0.314 \\
\quad genotype 0 & 698 (71.4\%) & 139 (67.5\%) & 559 (72.5\%) &  \\
\quad genotype 1 & 240 (24.6\%) & 58 (28.2\%) & 182 (23.6\%) &  \\
\quad genotype 2 & 27 (2.8\%) & 7 (3.4\%) & 20 (2.6\%) &  \\
\quad Missing, n (\%) & 12 (1.2\%) & 2 (1.0\%) & 10 (1.3\%) &  \\
\textit{rs113638840 (G)}, n (\%) &  &  &  & 0.352 \\
\quad genotype 0 & 520 (53.2\%) & 98 (47.6\%) & 422 (54.7\%) &  \\
\quad genotype 1 & 362 (37.1\%) & 80 (38.8\%) & 282 (36.6\%) &  \\
\quad genotype 2 & 52 (5.3\%) & 13 (6.3\%) & 39 (5.1\%) &  \\
\quad Missing, n (\%) & 43 (4.4\%) & 15 (7.3\%) & 28 (3.6\%) &  \\
\textit{rs72811372 (A)}, n (\%) &  &  &  & 0.008 \\
\quad genotype 0 & 801 (82.0\%) & 162 (78.6\%) & 639 (82.9\%) &  \\
\quad genotype 1 & 169 (17.3\%) & 40 (19.4\%) & 129 (16.7\%) &  \\
\quad genotype 2 & 3 (0.3\%) & 3 (1.5\%) & 0 (0.0\%) &  \\
\quad Missing, n (\%) & 4 (0.4\%) & 1 (0.5\%) & 3 (0.4\%) &  \\
\textit{rs7443323 (A)}, n (\%) &  &  &  & 0.350 \\
\quad genotype 0 & 537 (55.0\%) & 119 (57.8\%) & 418 (54.2\%) &  \\
\quad genotype 1 & 365 (37.4\%) & 75 (36.4\%) & 290 (37.6\%) &  \\
\quad genotype 2 & 68 (7.0\%) & 10 (4.9\%) & 58 (7.5\%) &  \\
\quad Missing, n (\%) & 7 (0.7\%) & 2 (1.0\%) & 5 (0.6\%) &  \\
\textit{rs11745375 (T)}, n (\%) &  &  &  & 0.786 \\
\quad genotype 0 & 277 (28.4\%) & 56 (27.2\%) & 221 (28.7\%) &  \\
\quad genotype 1 & 496 (50.8\%) & 109 (52.9\%) & 387 (50.2\%) &  \\
\quad genotype 2 & 204 (20.9\%) & 41 (19.9\%) & 163 (21.1\%) &  \\
\textit{rs55993676 (T)}, n (\%) &  &  &  & 0.275 \\
\quad genotype 0 & 504 (51.6\%) & 112 (54.4\%) & 392 (50.8\%) &  \\
\quad genotype 1 & 389 (39.8\%) & 82 (39.8\%) & 307 (39.8\%) &  \\
\quad genotype 2 & 66 (6.8\%) & 9 (4.4\%) & 57 (7.4\%) &  \\
\quad Missing, n (\%) & 18 (1.8\%) & 3 (1.5\%) & 15 (1.9\%) &  \\
\textit{rs196025 (A)}, n (\%) &  &  &  & 0.394 \\
\quad genotype 0 & 358 (36.6\%) & 68 (33.0\%) & 290 (37.6\%) &  \\
\quad genotype 1 & 436 (44.6\%) & 94 (45.6\%) & 342 (44.4\%) &  \\
\quad genotype 2 & 140 (14.3\%) & 34 (16.5\%) & 106 (13.7\%) &  \\
\quad Missing, n (\%) & 43 (4.4\%) & 10 (4.9\%) & 33 (4.3\%) &  \\
\textit{rs9393688 (T)}, n (\%) &  &  &  & 0.051 \\
\quad genotype 0 & 495 (50.7\%) & 113 (54.9\%) & 382 (49.5\%) &  \\
\quad genotype 1 & 372 (38.1\%) & 65 (31.6\%) & 307 (39.8\%) &  \\
\quad genotype 2 & 71 (7.3\%) & 20 (9.7\%) & 51 (6.6\%) &  \\
\quad Missing, n (\%) & 39 (4.0\%) & 8 (3.9\%) & 31 (4.0\%) &  \\
\textit{rs58453446 (C)}, n (\%) &  &  &  & 0.563 \\
\quad genotype 0 & 414 (42.4\%) & 94 (45.6\%) & 320 (41.5\%) &  \\
\quad genotype 1 & 451 (46.2\%) & 89 (43.2\%) & 362 (47.0\%) &  \\
\quad genotype 2 & 110 (11.3\%) & 23 (11.2\%) & 87 (11.3\%) &  \\
\quad Missing, n (\%) & 2 (0.2\%) & 0 (0.0\%) & 2 (0.3\%) &  \\
\textit{rs28719767 (C)}, n (\%) &  &  &  & 0.295 \\
\quad genotype 0 & 481 (49.2\%) & 111 (53.9\%) & 370 (48.0\%) &  \\
\quad genotype 1 & 394 (40.3\%) & 74 (35.9\%) & 320 (41.5\%) &  \\
\quad genotype 2 & 81 (8.3\%) & 18 (8.7\%) & 63 (8.2\%) &  \\
\quad Missing, n (\%) & 21 (2.1\%) & 3 (1.5\%) & 18 (2.3\%) &  \\
\textit{rs28517513 (T)}, n (\%) &  &  &  & 0.970 \\
\quad genotype 0 & 568 (58.1\%) & 117 (56.8\%) & 451 (58.5\%) &  \\
\quad genotype 1 & 362 (37.1\%) & 77 (37.4\%) & 285 (37.0\%) &  \\
\quad genotype 2 & 43 (4.4\%) & 9 (4.4\%) & 34 (4.4\%) &  \\
\quad Missing, n (\%) & 4 (0.4\%) & 3 (1.5\%) & 1 (0.1\%) &  \\
\textit{rs1528624 (G)}, n (\%) &  &  &  & 0.912 \\
\quad genotype 0 & 258 (26.4\%) & 54 (26.2\%) & 204 (26.5\%) &  \\
\quad genotype 1 & 454 (46.5\%) & 93 (45.1\%) & 361 (46.8\%) &  \\
\quad genotype 2 & 247 (25.3\%) & 54 (26.2\%) & 193 (25.0\%) &  \\
\quad Missing, n (\%) & 18 (1.8\%) & 5 (2.4\%) & 13 (1.7\%) &  \\
\textit{rs1425794 (T)}, n (\%) &  &  &  & 0.747 \\
\quad genotype 0 & 330 (33.8\%) & 69 (33.5\%) & 261 (33.9\%) &  \\
\quad genotype 1 & 465 (47.6\%) & 102 (49.5\%) & 363 (47.1\%) &  \\
\quad genotype 2 & 182 (18.6\%) & 35 (17.0\%) & 147 (19.1\%) &  \\
\textit{rs17387279 (G)}, n (\%) &  &  &  & 0.446 \\
\quad genotype 0 & 654 (66.9\%) & 139 (67.5\%) & 515 (66.8\%) &  \\
\quad genotype 1 & 271 (27.7\%) & 60 (29.1\%) & 211 (27.4\%) &  \\
\quad genotype 2 & 38 (3.9\%) & 5 (2.4\%) & 33 (4.3\%) &  \\
\quad Missing, n (\%) & 14 (1.4\%) & 2 (1.0\%) & 12 (1.6\%) &  \\
\textit{rs659398 (T)}, n (\%) &  &  &  & 0.247 \\
\quad genotype 0 & 499 (51.1\%) & 101 (49.0\%) & 398 (51.6\%) &  \\
\quad genotype 1 & 368 (37.7\%) & 77 (37.4\%) & 291 (37.7\%) &  \\
\quad genotype 2 & 73 (7.5\%) & 21 (10.2\%) & 52 (6.7\%) &  \\
\quad Missing, n (\%) & 37 (3.8\%) & 7 (3.4\%) & 30 (3.9\%) &  \\
\textit{rs10987386 (T)}, n (\%) &  &  &  & 0.697 \\
\quad genotype 0 & 674 (69.0\%) & 146 (70.9\%) & 528 (68.5\%) &  \\
\quad genotype 1 & 278 (28.5\%) & 54 (26.2\%) & 224 (29.1\%) &  \\
\quad genotype 2 & 25 (2.6\%) & 6 (2.9\%) & 19 (2.5\%) &  \\
\textit{rs4948502 (C)}, n (\%) &  &  &  & 0.085 \\
\quad genotype 0 & 353 (36.1\%) & 83 (40.3\%) & 270 (35.0\%) &  \\
\quad genotype 1 & 465 (47.6\%) & 84 (40.8\%) & 381 (49.4\%) &  \\
\quad genotype 2 & 159 (16.3\%) & 39 (18.9\%) & 120 (15.6\%) &  \\
\textit{rs12571363 (A)}, n (\%) &  &  &  & 0.900 \\
\quad genotype 0 & 791 (81.0\%) & 165 (80.1\%) & 626 (81.2\%) &  \\
\quad genotype 1 & 176 (18.0\%) & 39 (18.9\%) & 137 (17.8\%) &  \\
\quad genotype 2 & 10 (1.0\%) & 2 (1.0\%) & 8 (1.0\%) &  \\
\textit{rs7927422 (C)}, n (\%) &  &  &  & 0.410 \\
\quad genotype 0 & 349 (35.7\%) & 65 (31.6\%) & 284 (36.8\%) &  \\
\quad genotype 1 & 445 (45.5\%) & 99 (48.1\%) & 346 (44.9\%) &  \\
\quad genotype 2 & 161 (16.5\%) & 36 (17.5\%) & 125 (16.2\%) &  \\
\quad Missing, n (\%) & 22 (2.3\%) & 6 (2.9\%) & 16 (2.1\%) &  \\
\textit{rs11022690 (C)}, n (\%) &  &  &  & 0.374 \\
\quad genotype 0 & 280 (28.7\%) & 53 (25.7\%) & 227 (29.4\%) &  \\
\quad genotype 1 & 467 (47.8\%) & 100 (48.5\%) & 367 (47.6\%) &  \\
\quad genotype 2 & 220 (22.5\%) & 53 (25.7\%) & 167 (21.7\%) &  \\
\quad Missing, n (\%) & 10 (1.0\%) & 0 (0.0\%) & 10 (1.3\%) &  \\
\textit{rs11227223 (T)}, n (\%) &  &  &  & 0.107 \\
\quad genotype 0 & 871 (89.2\%) & 188 (91.3\%) & 683 (88.6\%) &  \\
\quad genotype 1 & 96 (9.8\%) & 14 (6.8\%) & 82 (10.6\%) &  \\
\quad genotype 2 & 2 (0.2\%) & 1 (0.5\%) & 1 (0.1\%) &  \\
\quad Missing, n (\%) & 8 (0.8\%) & 3 (1.5\%) & 5 (0.6\%) &  \\
\textit{rs12313454 (G)}, n (\%) &  &  &  & 0.468 \\
\quad genotype 0 & 735 (75.2\%) & 152 (73.8\%) & 583 (75.6\%) &  \\
\quad genotype 1 & 225 (23.0\%) & 52 (25.2\%) & 173 (22.4\%) &  \\
\quad genotype 2 & 12 (1.2\%) & 1 (0.5\%) & 11 (1.4\%) &  \\
\quad Missing, n (\%) & 5 (0.5\%) & 1 (0.5\%) & 4 (0.5\%) &  \\
\textit{rs10878300 (T)}, n (\%) &  &  &  & 0.819 \\
\quad genotype 0 & 694 (71.0\%) & 144 (69.9\%) & 550 (71.3\%) &  \\
\quad genotype 1 & 247 (25.3\%) & 53 (25.7\%) & 194 (25.2\%) &  \\
\quad genotype 2 & 36 (3.7\%) & 9 (4.4\%) & 27 (3.5\%) &  \\
\textit{rs4444235 (T)}, n (\%) &  &  &  & 0.200 \\
\quad genotype 0 & 256 (26.2\%) & 44 (21.4\%) & 212 (27.5\%) &  \\
\quad genotype 1 & 486 (49.7\%) & 108 (52.4\%) & 378 (49.0\%) &  \\
\quad genotype 2 & 235 (24.1\%) & 54 (26.2\%) & 181 (23.5\%) &  \\
\textit{rs1956028 (C)}, n (\%) &  &  &  & 0.942 \\
\quad genotype 0 & 739 (75.6\%) & 157 (76.2\%) & 582 (75.5\%) &  \\
\quad genotype 1 & 222 (22.7\%) & 45 (21.8\%) & 177 (23.0\%) &  \\
\quad genotype 2 & 14 (1.4\%) & 3 (1.5\%) & 11 (1.4\%) &  \\
\quad Missing, n (\%) & 2 (0.2\%) & 1 (0.5\%) & 1 (0.1\%) &  \\
\textit{rs72743477 (G)}, n (\%) &  &  &  & 0.744 \\
\quad genotype 0 & 608 (62.2\%) & 124 (60.2\%) & 484 (62.8\%) &  \\
\quad genotype 1 & 320 (32.8\%) & 72 (35.0\%) & 248 (32.2\%) &  \\
\quad genotype 2 & 45 (4.6\%) & 9 (4.4\%) & 36 (4.7\%) &  \\
\quad Missing, n (\%) & 4 (0.4\%) & 1 (0.5\%) & 3 (0.4\%) &  \\
\textit{rs4886509 (C)}, n (\%) &  &  &  & 0.677 \\
\quad genotype 0 & 422 (43.2\%) & 92 (44.7\%) & 330 (42.8\%) &  \\
\quad genotype 1 & 417 (42.7\%) & 83 (40.3\%) & 334 (43.3\%) &  \\
\quad genotype 2 & 103 (10.5\%) & 24 (11.7\%) & 79 (10.2\%) &  \\
\quad Missing, n (\%) & 35 (3.6\%) & 7 (3.4\%) & 28 (3.6\%) &  \\
\textit{rs4889526 (A)}, n (\%) &  &  &  & 0.611 \\
\quad genotype 0 & 417 (42.7\%) & 94 (45.6\%) & 323 (41.9\%) &  \\
\quad genotype 1 & 439 (44.9\%) & 87 (42.2\%) & 352 (45.7\%) &  \\
\quad genotype 2 & 114 (11.7\%) & 25 (12.1\%) & 89 (11.5\%) &  \\
\quad Missing, n (\%) & 7 (0.7\%) & 0 (0.0\%) & 7 (0.9\%) &  \\
\textit{rs7196853 (C)}, n (\%) &  &  &  & 0.683 \\
\quad genotype 0 & 773 (79.1\%) & 157 (76.2\%) & 616 (79.9\%) &  \\
\quad genotype 1 & 170 (17.4\%) & 40 (19.4\%) & 130 (16.9\%) &  \\
\quad genotype 2 & 6 (0.6\%) & 1 (0.5\%) & 5 (0.6\%) &  \\
\quad Missing, n (\%) & 28 (2.9\%) & 8 (3.9\%) & 20 (2.6\%) &  \\
\textit{rs750739 (G)}, n (\%) &  &  &  & 0.628 \\
\quad genotype 0 & 716 (73.3\%) & 151 (73.3\%) & 565 (73.3\%) &  \\
\quad genotype 1 & 209 (21.4\%) & 41 (19.9\%) & 168 (21.8\%) &  \\
\quad genotype 2 & 18 (1.8\%) & 5 (2.4\%) & 13 (1.7\%) &  \\
\quad Missing, n (\%) & 34 (3.5\%) & 9 (4.4\%) & 25 (3.2\%) &  \\
\textit{rs4800410 (C)}, n (\%) &  &  &  & 0.299 \\
\quad genotype 0 & 358 (36.6\%) & 83 (40.3\%) & 275 (35.7\%) &  \\
\quad genotype 1 & 444 (45.4\%) & 84 (40.8\%) & 360 (46.7\%) &  \\
\quad genotype 2 & 132 (13.5\%) & 30 (14.6\%) & 102 (13.2\%) &  \\
\quad Missing, n (\%) & 43 (4.4\%) & 9 (4.4\%) & 34 (4.4\%) &  \\
\textit{rs4636990 (G)}, n (\%) &  &  &  & 0.890 \\
\quad genotype 0 & 309 (31.6\%) & 68 (33.0\%) & 241 (31.3\%) &  \\
\quad genotype 1 & 434 (44.4\%) & 90 (43.7\%) & 344 (44.6\%) &  \\
\quad genotype 2 & 229 (23.4\%) & 47 (22.8\%) & 182 (23.6\%) &  \\
\quad Missing, n (\%) & 5 (0.5\%) & 1 (0.5\%) & 4 (0.5\%) &  \\
\textit{rs1548029 (C)}, n (\%) &  &  &  & 0.687 \\
\quad genotype 0 & 383 (39.2\%) & 82 (39.8\%) & 301 (39.0\%) &  \\
\quad genotype 1 & 453 (46.4\%) & 91 (44.2\%) & 362 (47.0\%) &  \\
\quad genotype 2 & 141 (14.4\%) & 33 (16.0\%) & 108 (14.0\%) &  \\
\textit{rs72490631 (C)}, n (\%) &  &  &  & 0.091 \\
\quad genotype 0 & 592 (60.6\%) & 113 (54.9\%) & 479 (62.1\%) &  \\
\quad genotype 1 & 331 (33.9\%) & 83 (40.3\%) & 248 (32.2\%) &  \\
\quad genotype 2 & 49 (5.0\%) & 9 (4.4\%) & 40 (5.2\%) &  \\
\quad Missing, n (\%) & 5 (0.5\%) & 1 (0.5\%) & 4 (0.5\%) &  \\
\textit{rs6006399 (G)}, n (\%) &  &  &  & 0.479 \\
\quad genotype 0 & 801 (82.0\%) & 175 (85.0\%) & 626 (81.2\%) &  \\
\quad genotype 1 & 164 (16.8\%) & 29 (14.1\%) & 135 (17.5\%) &  \\
\quad genotype 2 & 9 (0.9\%) & 2 (1.0\%) & 7 (0.9\%) &  \\
\quad Missing, n (\%) & 3 (0.3\%) & 0 (0.0\%) & 3 (0.4\%) &  \\
\end{longtable}
\endgroup

\begingroup\footnotesize
\setlength{\tabcolsep}{4pt}
\renewcommand{\arraystretch}{1.0}

\begin{longtable}{rr l r r r r r r}
\caption{Simulation results for the five non-zero covariates, averaged over
the replicates indicated in column $n_{\mathrm{valid}}$.
Reported quantities are absolute bias ($|\mathrm{Bias}|$), root mean squared
error (RMSE), empirical standard error (EmpSE), the average model-based
standard error (AvgSE) and the empirical coverage probability of the nominal
95\% confidence intervals (Cov).}
\label{tab:sim_nonzero_full}\\

\toprule
$n$ & $p$ & Method & $n_{\mathrm{valid}}$ & $|\mathrm{Bias}|$ & RMSE & EmpSE & AvgSE & Cov \\
\midrule
\endfirsthead

\multicolumn{9}{l}{\emph{Table \ref{tab:sim_nonzero_full} continued from previous page}}\\
\toprule
$n$ & $p$ & Method & $n_{\mathrm{valid}}$ & $|\mathrm{Bias}|$ & RMSE & EmpSE & AvgSE & Cov \\
\midrule
\endhead

\midrule
\multicolumn{9}{r}{\emph{continued on next page}}\\
\endfoot

\bottomrule
\endlastfoot

500  & 20  & Oracle           & 100 & 0.018 & 0.082 & 0.081 & 0.070 & 0.898 \\
     &     & SMC-DBL          & 100 & 0.013 & 0.088 & 0.088 & 0.084 & 0.938 \\
     &     & Standard IRO     & 100 & 0.087 & 0.121 & 0.083 & 0.075 & 0.748 \\
     &     & Std SMC-FCS      & 100 & 0.025 & 0.099 & 0.096 & 0.087 & 0.930 \\
     &     & Mean Imp + DL    & 100 & 0.052 & 0.096 & 0.080 & 0.071 & 0.840 \\
\cmidrule(lr){2-9}
     & 50  & Oracle           & 100 & 0.025 & 0.083 & 0.078 & 0.069 & 0.916 \\
     &     & SMC-DBL          & 100 & 0.014 & 0.087 & 0.084 & 0.084 & 0.946 \\
     &     & Standard IRO     & 100 & 0.133 & 0.154 & 0.076 & 0.079 & 0.608 \\
     &     & Std SMC-FCS      & 100 & 0.054 & 0.118 & 0.104 & 0.091 & 0.884 \\
     &     & Mean Imp + DL    & 100 & 0.047 & 0.096 & 0.083 & 0.070 & 0.852 \\
\cmidrule(lr){2-9}
     & 100 & Oracle           & 100 & 0.039 & 0.096 & 0.086 & 0.070 & 0.854 \\
     &     & SMC-DBL          & 100 & 0.031 & 0.099 & 0.093 & 0.085 & 0.898 \\
     &     & Standard IRO     & 100 & 0.184 & 0.201 & 0.081 & 0.090 & 0.470 \\
     &     & Std SMC-FCS      & 100 & 0.124 & 0.200 & 0.155 & 0.113 & 0.790 \\
     &     & Mean Imp + DL    & 100 & 0.041 & 0.100 & 0.088 & 0.071 & 0.836 \\
\cmidrule(lr){2-9}
     & 200 & Oracle           & 100 & 0.063 & 0.108 & 0.083 & 0.071 & 0.802 \\
     &     & SMC-DBL          & 100 & 0.041 & 0.104 & 0.092 & 0.088 & 0.908 \\
     &     & Standard IRO     & 100 & 0.350 & 0.359 & 0.077 & 0.160 & 0.412 \\
     &     & Std SMC-FCS      & 14  & 5.354 & 5.639 & 1.680 & 2.580 & 0.471 \\
     &     & Mean Imp + DL    & 100 & 0.034 & 0.097 & 0.089 & 0.072 & 0.844 \\
\midrule
1000 & 20  & Oracle           & 100 & 0.006 & 0.050 & 0.049 & 0.049 & 0.952 \\
     &     & SMC-DBL          & 100 & 0.017 & 0.060 & 0.057 & 0.059 & 0.950 \\
     &     & Standard IRO     & 100 & 0.061 & 0.083 & 0.056 & 0.051 & 0.722 \\
     &     & Std SMC-FCS      & 100 & 0.011 & 0.061 & 0.060 & 0.061 & 0.948 \\
     &     & Mean Imp + DL    & 100 & 0.067 & 0.087 & 0.054 & 0.050 & 0.712 \\
\cmidrule(lr){2-9}
     & 50  & Oracle           & 100 & 0.012 & 0.055 & 0.054 & 0.049 & 0.926 \\
     &     & SMC-DBL          & 100 & 0.021 & 0.066 & 0.062 & 0.059 & 0.916 \\
     &     & Standard IRO     & 100 & 0.100 & 0.116 & 0.058 & 0.053 & 0.528 \\
     &     & Std SMC-FCS      & 100 & 0.021 & 0.074 & 0.071 & 0.062 & 0.902 \\
     &     & Mean Imp + DL    & 100 & 0.056 & 0.084 & 0.060 & 0.050 & 0.740 \\
\cmidrule(lr){2-9}
     & 100 & Oracle           & 100 & 0.020 & 0.058 & 0.054 & 0.049 & 0.910 \\
     &     & SMC-DBL          & 100 & 0.019 & 0.064 & 0.060 & 0.059 & 0.926 \\
     &     & Standard IRO     & 100 & 0.129 & 0.140 & 0.054 & 0.056 & 0.366 \\
     &     & Std SMC-FCS      & 100 & 0.040 & 0.083 & 0.072 & 0.065 & 0.880 \\
     &     & Mean Imp + DL    & 100 & 0.053 & 0.078 & 0.056 & 0.050 & 0.754 \\
\cmidrule(lr){2-9}
     & 200 & Oracle           & 100 & 0.034 & 0.066 & 0.055 & 0.049 & 0.850 \\
     &     & SMC-DBL          & 100 & 0.020 & 0.062 & 0.057 & 0.060 & 0.956 \\
     &     & Standard IRO     & 100 & 0.173 & 0.182 & 0.054 & 0.068 & 0.268 \\
     &     & Std SMC-FCS      & 100 & 0.103 & 0.139 & 0.091 & 0.078 & 0.752 \\
     &     & Mean Imp + DL    & 100 & 0.040 & 0.071 & 0.056 & 0.050 & 0.826 \\
\midrule
2000 & 20  & Oracle           & 100 & 0.005 & 0.036 & 0.036 & 0.035 & 0.944 \\
     &     & SMC-DBL          & 100 & 0.015 & 0.043 & 0.040 & 0.042 & 0.948 \\
     &     & Standard IRO     & 100 & 0.045 & 0.060 & 0.039 & 0.036 & 0.728 \\
     &     & Std SMC-FCS      & 100 & 0.004 & 0.042 & 0.042 & 0.043 & 0.960 \\
     &     & Mean Imp + DL    & 100 & 0.071 & 0.081 & 0.037 & 0.036 & 0.490 \\
\cmidrule(lr){2-9}
     & 50  & Oracle           & 100 & 0.010 & 0.040 & 0.039 & 0.035 & 0.918 \\
     &     & SMC-DBL          & 100 & 0.016 & 0.046 & 0.042 & 0.042 & 0.904 \\
     &     & Standard IRO     & 100 & 0.067 & 0.079 & 0.041 & 0.036 & 0.568 \\
     &     & Std SMC-FCS      & 100 & 0.013 & 0.048 & 0.046 & 0.043 & 0.932 \\
     &     & Mean Imp + DL    & 100 & 0.066 & 0.077 & 0.039 & 0.036 & 0.554 \\
\cmidrule(lr){2-9}
     & 100 & Oracle           & 100 & 0.011 & 0.041 & 0.039 & 0.035 & 0.922 \\
     &     & SMC-DBL          & 100 & 0.017 & 0.048 & 0.044 & 0.042 & 0.912 \\
     &     & Standard IRO     & 100 & 0.090 & 0.099 & 0.040 & 0.037 & 0.374 \\
     &     & Std SMC-FCS      & 100 & 0.019 & 0.052 & 0.048 & 0.045 & 0.918 \\
     &     & Mean Imp + DL    & 100 & 0.059 & 0.071 & 0.039 & 0.036 & 0.598 \\
\cmidrule(lr){2-9}
     & 200 & Oracle           & 100 & 0.025 & 0.048 & 0.040 & 0.035 & 0.852 \\
     &     & SMC-DBL          & 100 & 0.012 & 0.046 & 0.043 & 0.042 & 0.930 \\
     &     & Standard IRO     & 100 & 0.108 & 0.115 & 0.040 & 0.039 & 0.246 \\
     &     & Std SMC-FCS      & 100 & 0.046 & 0.069 & 0.050 & 0.046 & 0.800 \\
     &     & Mean Imp + DL    & 100 & 0.052 & 0.066 & 0.040 & 0.036 & 0.674 \\

\end{longtable}
\endgroup.
 
 

\end{document}